\documentclass[a4paper]{llncs}
\usepackage{algorithm}
\usepackage{algpseudocode}
\usepackage{ifpdf}
    \usepackage[pdftex]{graphicx}
    \usepackage{graphicx}

\usepackage{color}
\usepackage{amssymb}
\usepackage{arcs}
\algrenewcommand{\algorithmiccomment}[1]{/* #1 */}

\newcommand{\mt}[1]{{\tt #1}}
\newcommand{\cmd}[2]{\State \textbf{#1} #2}

\newcounter{step-ctr}

\newlength{\txtSize}
\newcommand{\hcomment}[2]{\setlength{\txtSize}{\hsize}
        \addtolength{\txtSize}{-7mm} \addtolength{\txtSize}{-#1}
        \Statex\hspace*{#1}\parbox[t]{\txtSize}{\Comment{#2}}
}
\newlength{\cmdLen}
\newcommand{\cmdp}[3]{\settowidth{\cmdLen}{#2}
        \setlength{\txtSize}{\hsize} \addtolength{\txtSize}{-#1}
        \addtolength{\txtSize}{-\cmdLen} \addtolength{\txtSize}{-9mm}
        \State \textbf{#2} \parbox[t]{\txtSize}{#3} 
}
\algloop[iIf]{iIf}
\algloopdefx[iIf]{iIf}[1]{\textbf{if} #1}%

\begin{document}

\title{Localizability of Wireless Sensor Networks:\\ Beyond Wheel Extension}
\author{Buddhadeb Sau\inst{1} \and Krishnendu Mukhopadhyaya\inst{2}}
\institute{Department of Mathematics, Jadavpur University, Kolkata, India,  
            \email{bsau@math.jdvu.ac.in}
            \and ACM Unit, Indian Statistical Institute, Kolkata, India, 
                \email{krishnendu@isical.ac.in}
            }


\date{}
\maketitle

\abovedisplayskip 0mm
\belowdisplayskip 0mm
\abovedisplayshortskip 0mm
\belowdisplayshortskip 0mm

\graphicspath{{figures/},{simulations/},{temp/}}

\begin{abstract}  
A network is called \textit{localizable} if the positions of all the nodes of the network can be 
computed uniquely. If a network is localizable and embedded in plane with \textit{generic 
configuration}, the positions of the nodes may be computed uniquely in finite time. Therefore,
identifying localizable networks is an important function. If the complete  information about the 
network is available at a single place, localizability can be tested in polynomial time. In a 
distributed environment, networks with \textit{trilateration ordering}s (popular in real 
applications) and \textit{wheel extension}s (a specific class of localizable networks) embedded in
plane can be identified by existing techniques. We propose a distributed technique which efficiently 
identifies a larger class of localizable networks. This class covers both trilateration and wheel 
extensions. In reality, exact distance is almost impossible or costly. The proposed algorithm 
based only on connectivity information. It requires no distance information. 
\end{abstract}


\section{Introduction}

A \textit{sensor} is a small sized and low powered electronic device with 
limited computational and communicating capability. A \textit{sensor network} is a network 
containing some ten to millions of sensors.  Wireless sensor networks (WSNs) have 
wide-ranging applications in problems such as traffic control, habitat monitoring, battlefield 
surveillance (e.g., intruder detection or giving assistance to mobile soldiers etc.), fire detection 
for monitoring forest-fires, disaster management, to alert the appearance of phytoplankton under the 
sea, etc. WSNs can help in gathering information from regions, where human access  is difficult. To 
react to an event detected by a sensor, the knowledge about the position of the sensor is necessary. 
Sensor deployment may be random. For example, they may be dropped from an air vehicle with no 
pre-defined infra-structure. In such cases, the positions of the sensors are completely unknown to 
start with. The problem of finding the positions of nodes in a network is known as \textit{network 
localization}.

A wireless ad-hoc network embedded in $\mathbb R^m$ ($m$-dimensional Euclidean space) may be 
represented by a \textit{distance graph} $G=(V,E,d)$ whose vertices represent computing devices. A 
pair of nodes, $\{u,v\} \in E$ if and only if the Euclidean distance between $u$ and $v$ ($d(u,v) = 
|u-v|$) is known. Determining the coordinates of vertices in an embedding of $G$ in $\mathbb R^m$ 
may be considered as graph realization problem~\cite{L70,H92,CHH02,AEG+06}. A \textit{realization} 
of a distance graph $G=(V,E,d)$ in $\mathbb R^m$ is an injective mapping $p:V \rightarrow \mathbb 
R^m$ such that $|p(u)-p(v)|=d(u,v)$, $\forall \{u,v\}\in E$ (i.e., one-to-one assignment of 
coordinates $(x_1,\ldots, x_m)\in \mathbb R^m$ to every vertex in $V$ so that the $d(u,v)$ 
represents the distance between $u$ and $v$. The pair $(G,p)$ is called a \textit{framework} of $G$ 
in $\mathbb{R}^m$. Two frameworks $(G,p)$ and $(G,q)$ are \textit{congruent} if 
$|p(u)-p(v)|=|q(u)-q(v)|$, $\forall u,v\in V$ (i.e., preserving the distances between all pairs of 
vertices). A framework $(G,p)$ is \textit{rigid}, if it has no smooth deformation~\cite{C05} 
preserving the edge lengths. The distance graph $G$ is \textit{generically globally 
rigid}~\cite{C05}, if all realizations with \textit{generic configuration}s (set of points with 
coordinates not being \textit{algebraically dependent}) are congruent. A set $A = \{\alpha_1, 
\ldots, \alpha_m \}$ of real numbers is algebraically dependent if there is a non-zero polynomial 
$h$ with integer coefficients such that $h(\alpha_1,\ldots,\alpha_m) =0$. The graph $G$ is termed as 
\textit{globally rigid}, if all realizations of $G$ are congruent. In this work, we consider 
realizations of a distance graph only in plane with generic configurations. Here onwards, 
all the discussions and results are concerned only in $\mathbb R^2$.

In real applications, positions of the nodes of a network may either be 1) 
estimated~\cite{MLRT04,BTY08,CE10,KSP10,ZSY10} within some tolerable error level or 2) uniquely 
realized. If all realizations of a network are congruent, the network is called 
\textit{localizable}. 
A distance graph is localizable if and only if it is globally rigid. Starting from 
three \textit{anchor}s (nodes with known unique position), all the nodes in a globally rigid graph 
can be uniquely realized. If the nodes cannot be localized using the given information, in several 
applications location estimation may serve well. Localization of a network is $NP$-hard~\cite{S79} 
even when it is localizable~\cite{JJ05}. Jackson and Jord\'{a}n~\cite{JJ05} proved that a graph is 
globally rigid, if and only if it is $3$-\textit{connected} and \textit{redundantly rigid}. A graph 
is $3$-connected, if at least three vertices must be removed to make the graph disconnected. A graph 
is redundantly rigid, if it remains rigid after removing any edge. Localizability of a graph can be 
answered in polynomial time~\cite{JJ05} by testing the 3-connectivity and redundant rigidity when 
complete network wide information is available in a single machine. To gather complete network 
information in a single machine is infeasible or very costly. On the contrary, finding methods to 
recognize globally rigid graphs in distributed environments based on local information is still a 
challenging problem~\cite{YLL10}.

The rest of the paper is organized as follows. Section~\ref{sec:contribution} describes 
motivation and contribution of this work. Section~\ref{sec:localizableGraphs} introduces
some classes of localizable graphs which include trilateration graph and wheel extension as their
special cases. Section~\ref{sec:problemAndMapping} defines the problem. It also describes a 
mapping of the problem into triangle bar recognition.
Section~\ref{sec:localizabilityTesting} describes the proposed distributed algorithm of
localizability testing. Section~\ref{sec:performanceAnalysis} proves  correctness and 
performance analysis of the algorithm. Finally, we conclude in Section~\ref{conclude}.

\section{Background and our contribution} \label{sec:contribution}
The most commonly used technique for localization is \textit{trilateration}~\cite{NN03,SHS01}. It 
efficiently localizes a \textit{trilateration graph} starting from three anchors. A trilateration 
graph is a graph with a \textit{trilateration ordering}, $\pi = (u_1, u_2,\ldots, u_n)$ where 
$u_1$, $u_2$, $u_3$ form a $K_3$ and every $u_i$ ($i > 3$) is adjacent to at least three nodes 
before $u_i$ in $\pi$. However, not all localizable networks admit trilateration ordering.
\begin{figure}[h]\vspace*{-4mm}
    \centering
    \begin{minipage}[t]{2.1cm}
     \includegraphics[height=.6in]{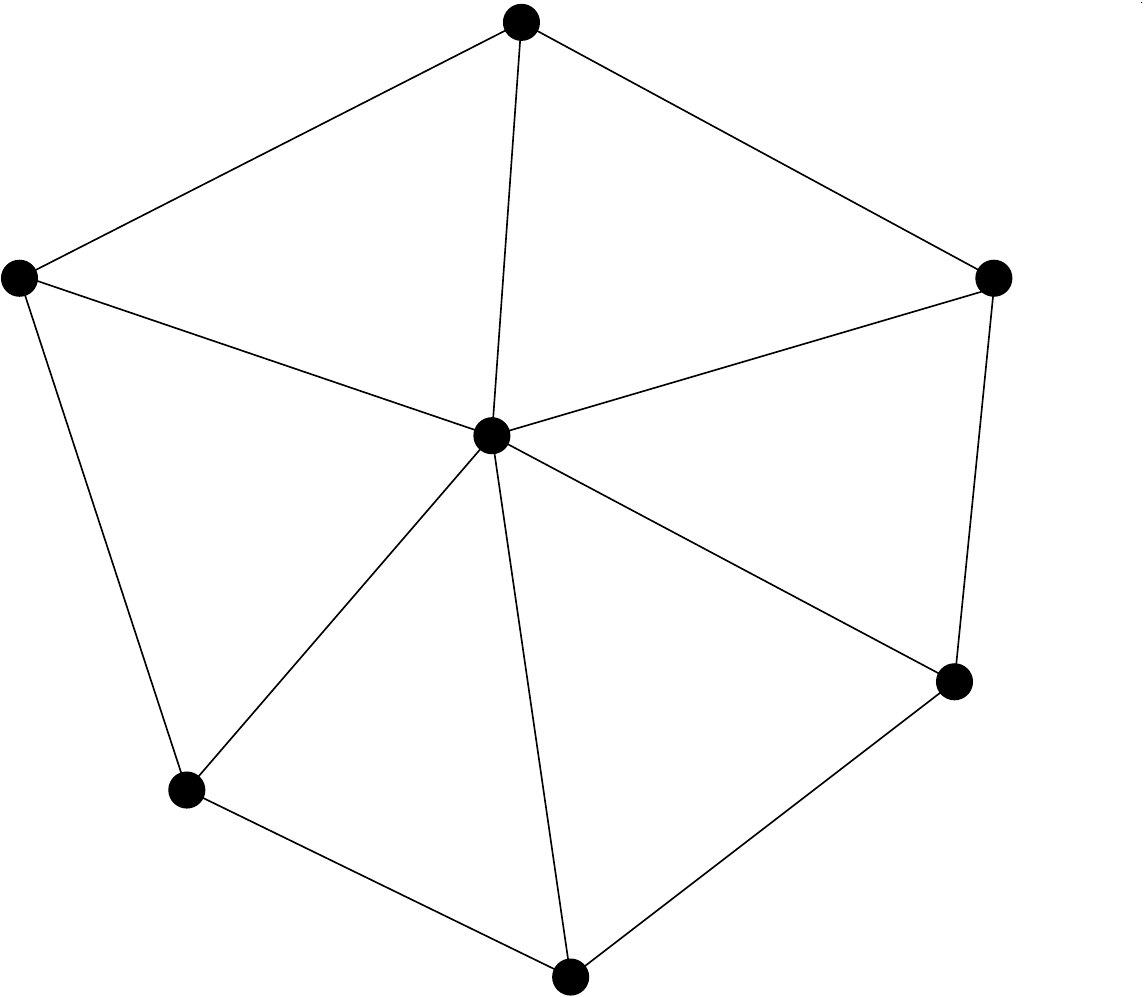}\\
      (a) \parbox[t]{1.5cm}{Wheel\\ graph}
    \end{minipage}
    \begin{minipage}[t]{2.3cm}
    \includegraphics[height=.8in]{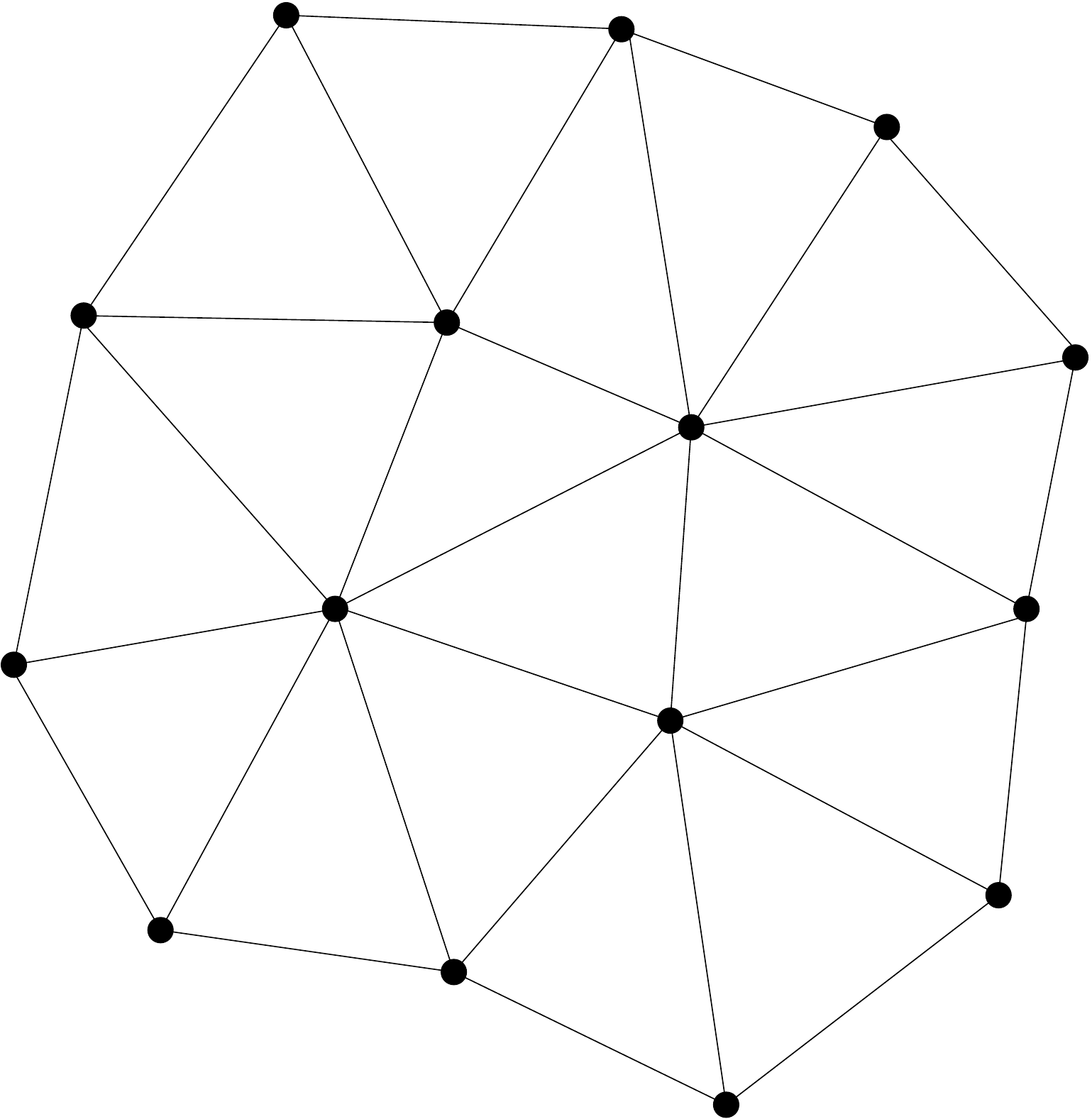} \\
    \parbox[t]{1.8cm}{(b) Wheel\\ extension}
    \end{minipage}
    \begin{minipage}[t]{2.3cm}
    \includegraphics[height=.8in]{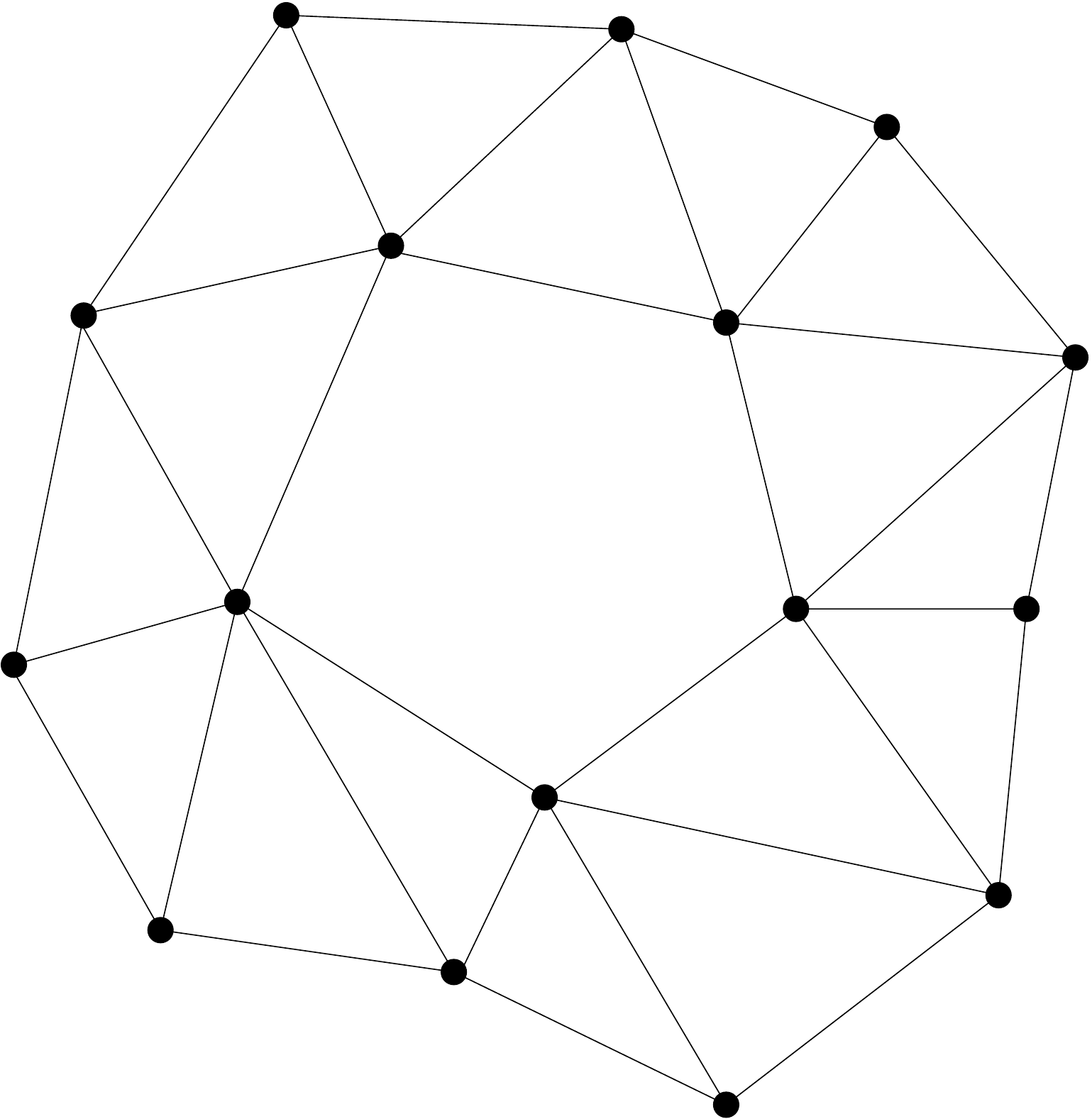}\\
    (c) \parbox[t]{1.8cm}{Triangle\\ cycle}
    \end{minipage}
    \begin{minipage}[t]{2.3cm}
    \includegraphics[height=.8in]{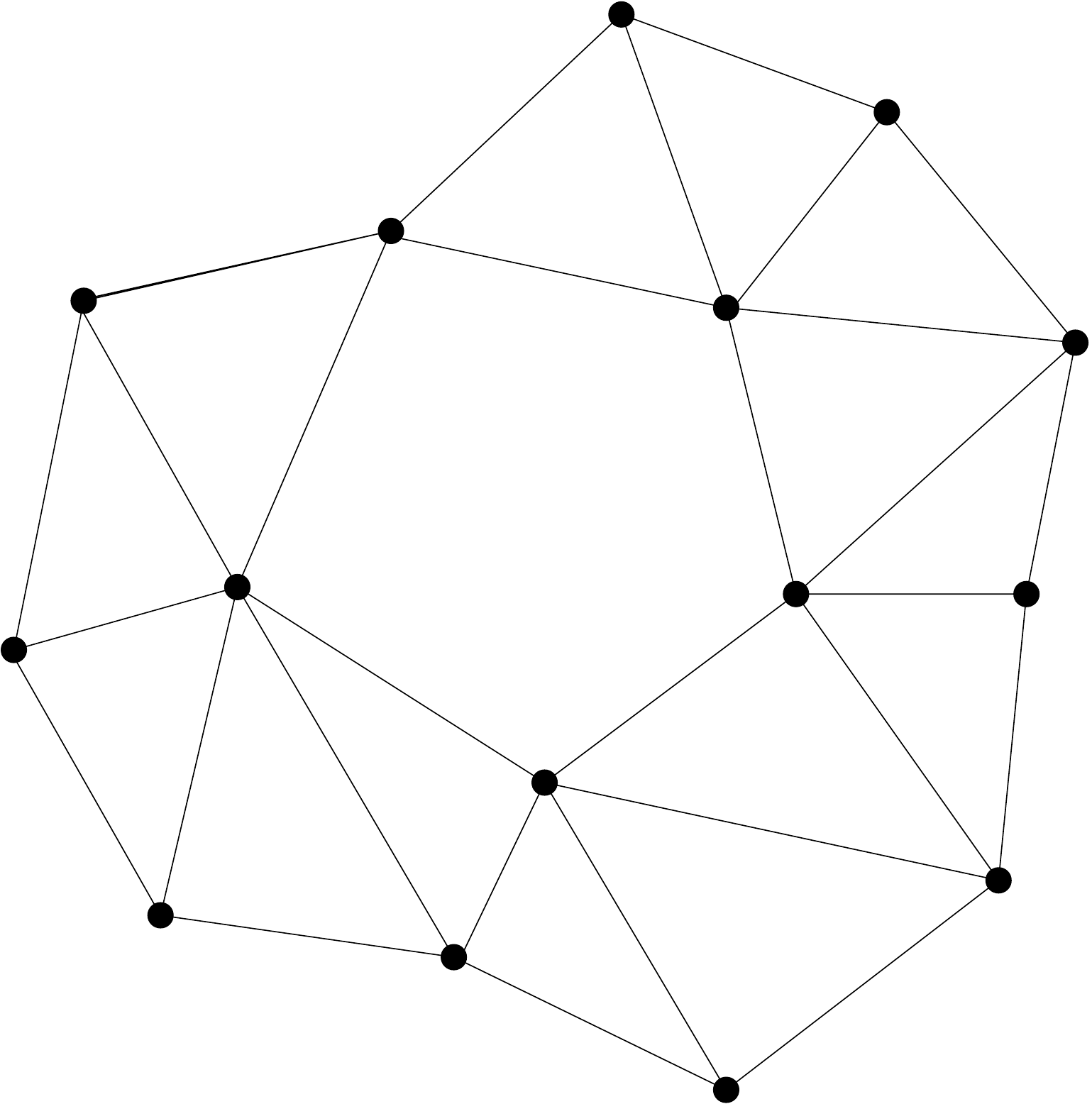}
        \parbox[t]{2.1cm}{(d) Triangle\\ circuit}
    \end{minipage}
    \begin{minipage}[t]{2.3cm}
    \includegraphics[height=.8in]{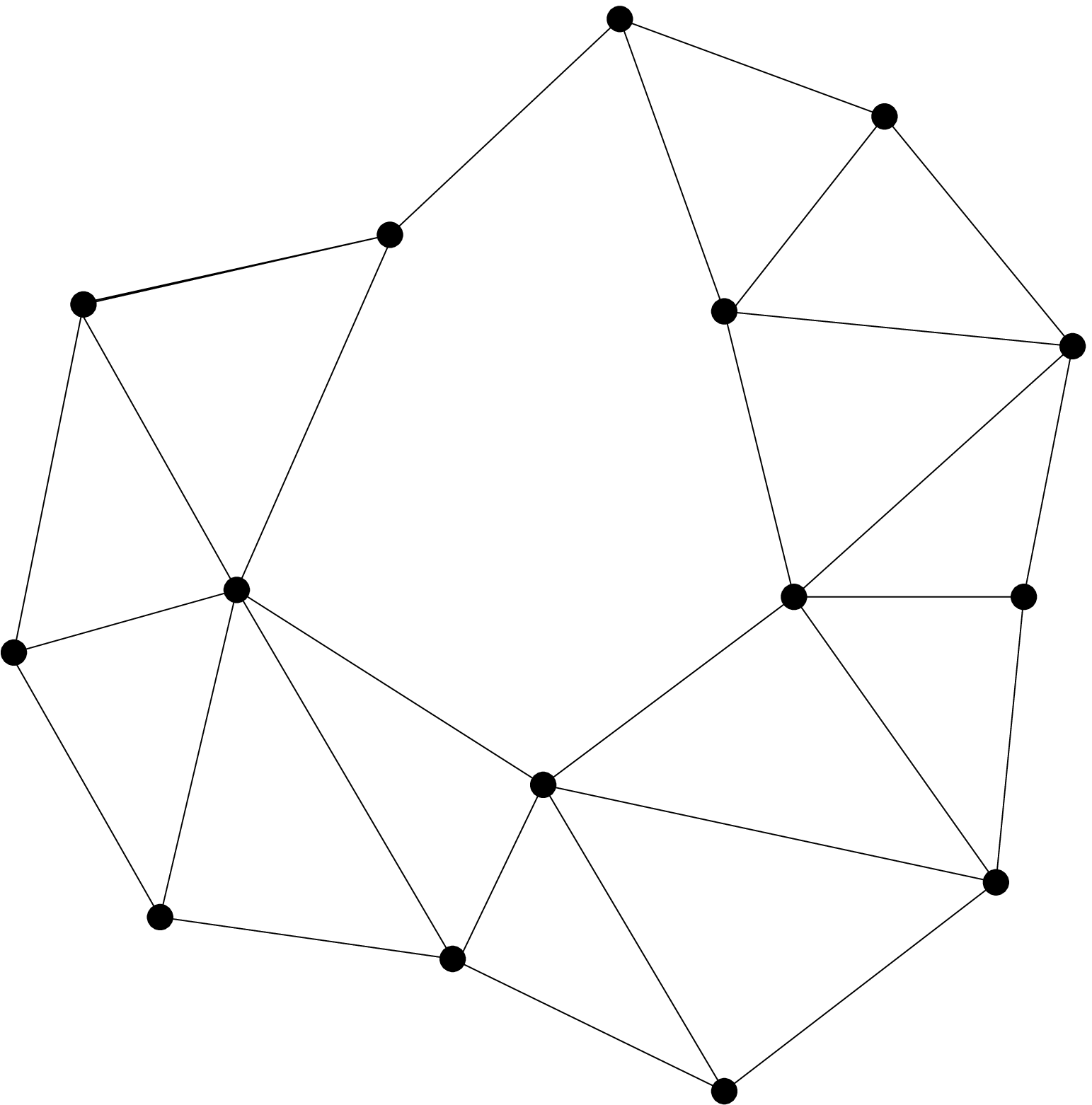}
    (e) \parbox[t]{1.8cm}{Triangle\\ bridge}
    \end{minipage}
    \begin{minipage}[t]{2cm}
    \includegraphics[height=.8in]{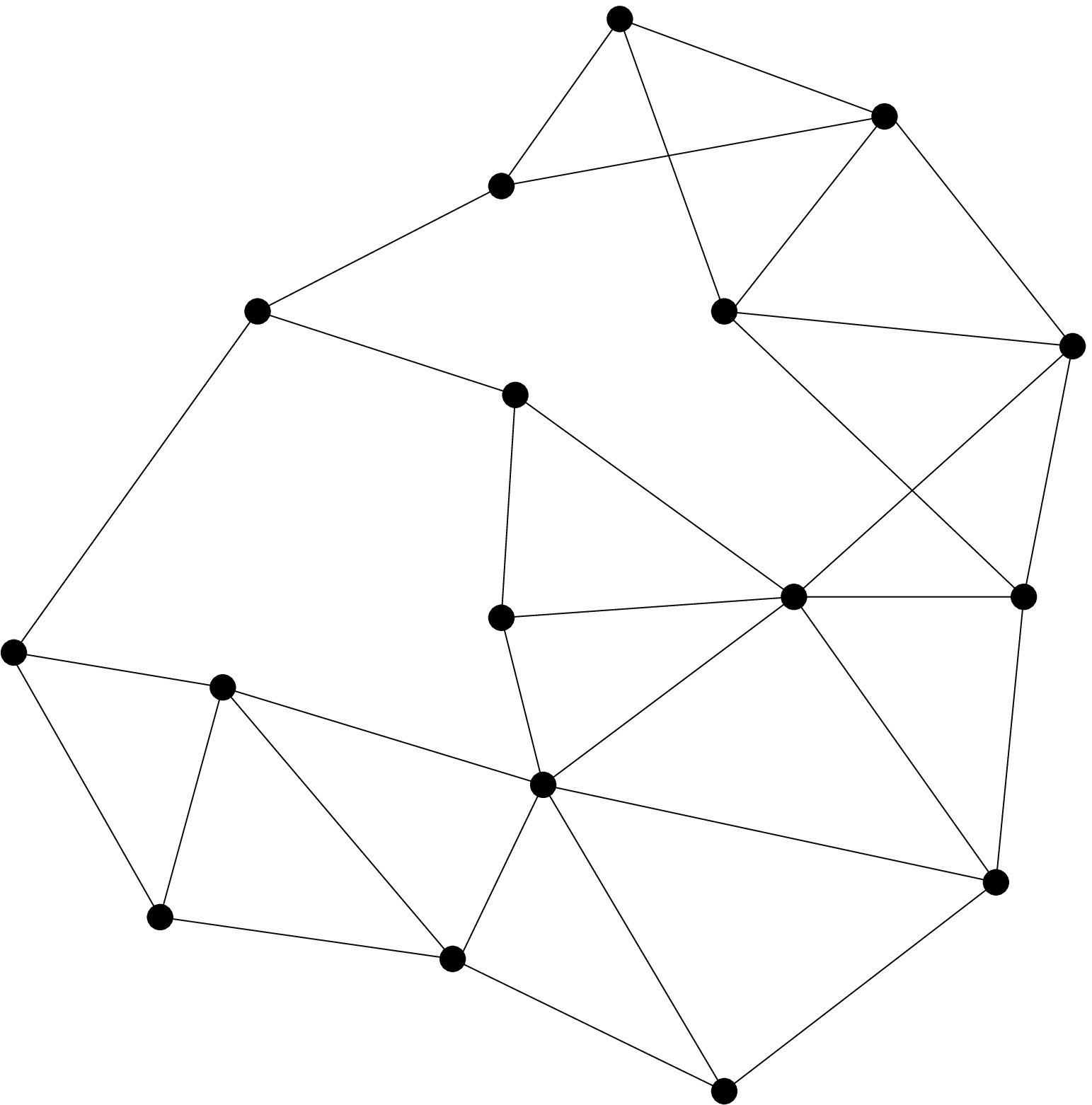}\\
    (f) \parbox[t]{1.5cm}{Triangle\\ net}
    \end{minipage}\\[-1mm]
    \caption{Examples of some localizable graphs having different properties}    
\label{fig:wheelGraphExtnCycle}
\vspace*{-5mm}
\end{figure}
Fig.~\ref{fig:wheelGraphExtnCycle}\,(a) and \ref{fig:wheelGraphExtnCycle}\,(b) respectively show a
\textit{wheel graph} and a \textit{wheel extension} graph which have no trilateration ordering.
A wheel $W_n$ with $n$ vertices is a graph consisting of a cycle with $n-1$
nodes and a vertex which is adjacent to all vertices on the cycle. A \textit{wheel
extension} is a graph having an ordering $\pi = (u_1, u_2, \ldots, u_n)$ of
nodes where $u_1$, $u_2$, $u_3$ form a $K_3$ and each $u_i$, $i > 3$, lies in
a wheel subgraph containing at least three nodes before $u_i$ in $\pi$. A wheel
extension is generically globally rigid and its localizability can be identified
efficiently and distributedly~\cite{YLL10}. However, there are many more localizable graphs which
do not have wheel extensions. For example,
Fig.~\ref{fig:wheelGraphExtnCycle}\,(c),
\ref{fig:wheelGraphExtnCycle}\,(d), \ref{fig:wheelGraphExtnCycle}\,(e)
and \ref{fig:wheelGraphExtnCycle}\,(f)
%
%
are examples of graphs which are generically globally rigid, but do not have wheel extensions.

The main contributions of this paper are as follows.  It introduces some elementary class 
of localizable graphs \textit{triangle cycle, triangle circuit, triangle bridge, triangle notch 
and triangle net}. Using these elementary classes of graphs, we build up a new family of 
generically globally rigid graphs called \textit{triangle bar}.  Trilateration graphs and wheel 
extensions are special cases of triangle bars. We propose an efficient distributed algorithm that 
recognizes triangle bars starting from a $K_3$ based only on connectivity information. It requires 
no distance information. In real applications, exact node distance is impossible or costly. 
However, several localizable graphs still fall outside the class triangle bar. To the best of our 
knowledge, distributedly recognizing an arbitrary localizable network still is an open problem.

\section{Rigidity and localizability of triangle bar}
\label{sec:localizableGraphs}

Unique realizability is closely related to graph rigidity~\cite{L70,H92}. The realizability testing 
of a distance graph $G=(V,E,d)$ is $NP$-hard~\cite{S79}. We expect data are consistent to have a 
realization, if the distance information is collected from an actual deployment of devices. A 
realization of $G$ may be visualized as a frame constructed by a finite set of hinged rods. The 
junctions and free ends are considered as vertices of the realization and rods as the edges. With 
perturbation on the frame, we may have a different realization preserving the edge distances. The 
realizations obtained by flipping, rotating or shifting the whole structure are congruent. By 
\textit{flip}, \textit{rotation} or \textit{shift} on a realization, we mean a part of the 
realization is flipped, rotated or shifted. If two globally rigid graphs in $\mathbb R^2$ share 
exactly one vertex in common, one of them may be rotated around the common vertex keeping the other 
fixed. Such a vertex is called a \textit{joint}. If two globally rigid graphs  in $\mathbb R^2$ 
share exactly two vertices, rotation about these vertices is no longer possible, but one of the 
graphs may be flipped, about the line joining the two common vertices, keeping the other fixed. This 
pair of vertices is called a \textit{flip}.
%
\begin{lemma}[\cite{SM09}] \label{lem:3commonPoint}
If two globally rigid subgraphs, $B_1$ and $B_2$, of a graph embedded in plane share at least three
non-collinear vertices, then $B_1 \cup B_2$ is globally rigid.
\end{lemma}
In this section, we formally introduce some elementary classes of localizable graphs: 
\textit{triangle cycle, triangle circuit, triangle bridge, triangle notch, triangle net}. 
Using these elementary classes, a larger class of localizable graphs \textit{triangle bar} is 
formally defined.

\subsection{Triangle cycle, triangle circuit and triangle bridge}
Let ${\cal T} = (T_1, T_2, \ldots, T_m)$ be a sequence of distinct triangles such that for 
every $i$, $2 \leq i \leq m-1$, $T_i$ shares two distinct edges with $T_{i-1}$ and $T_{i+1}$.
Such a sequence ${\cal T}$ of triangles is called a \textit{triangle stream}
(Fig.~\ref{fig:trianglerChainCycleCircuitBridge}\,(a)).
\begin{figure}[h]\vspace*{-4mm}
    \centering
    \begin{minipage}[t]{3cm}
     \includegraphics[height=.7in]{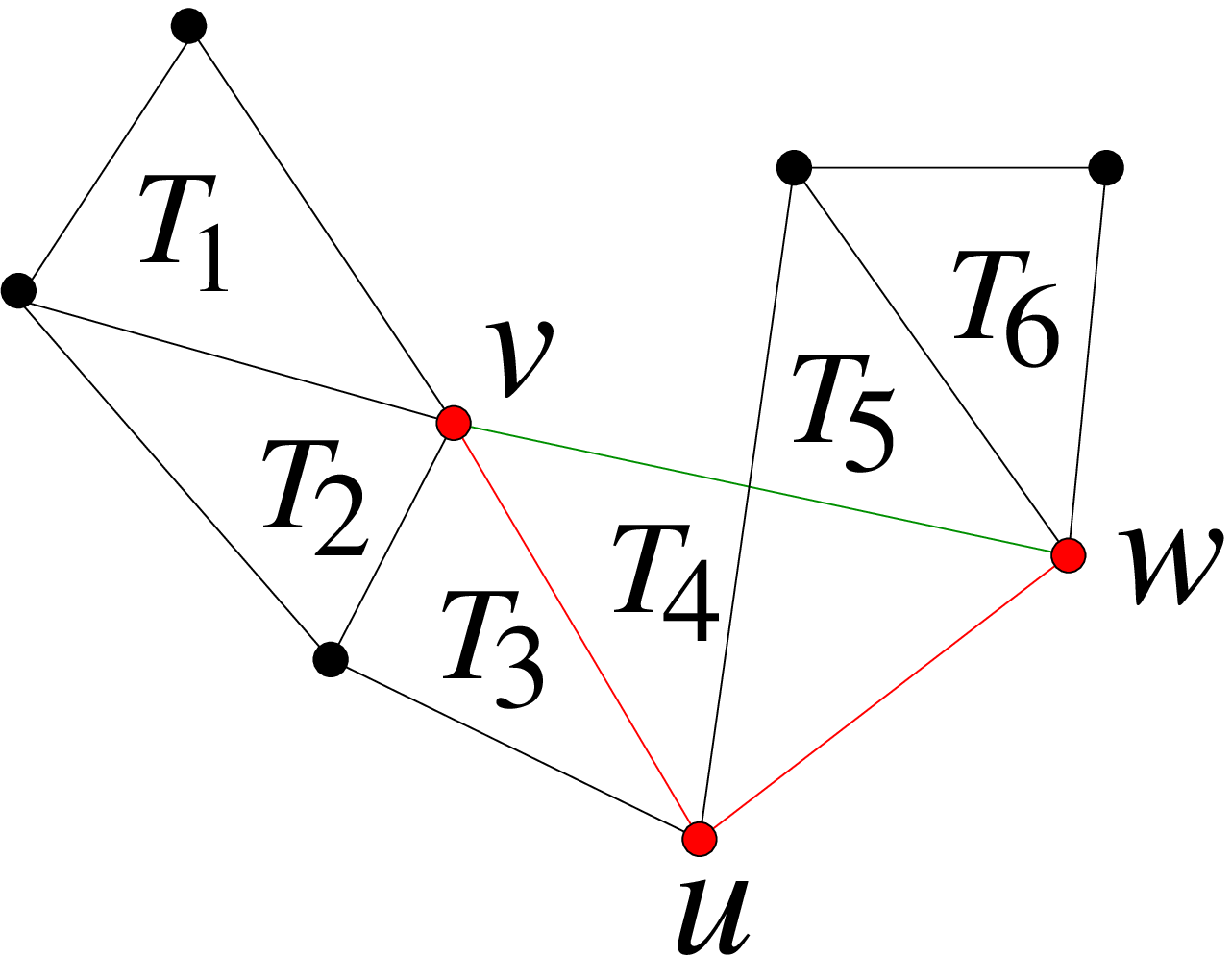}\\
     (a) Triangle chain
    \end{minipage}
    \begin{minipage}[t]{3cm}
    \includegraphics[height=.7in]{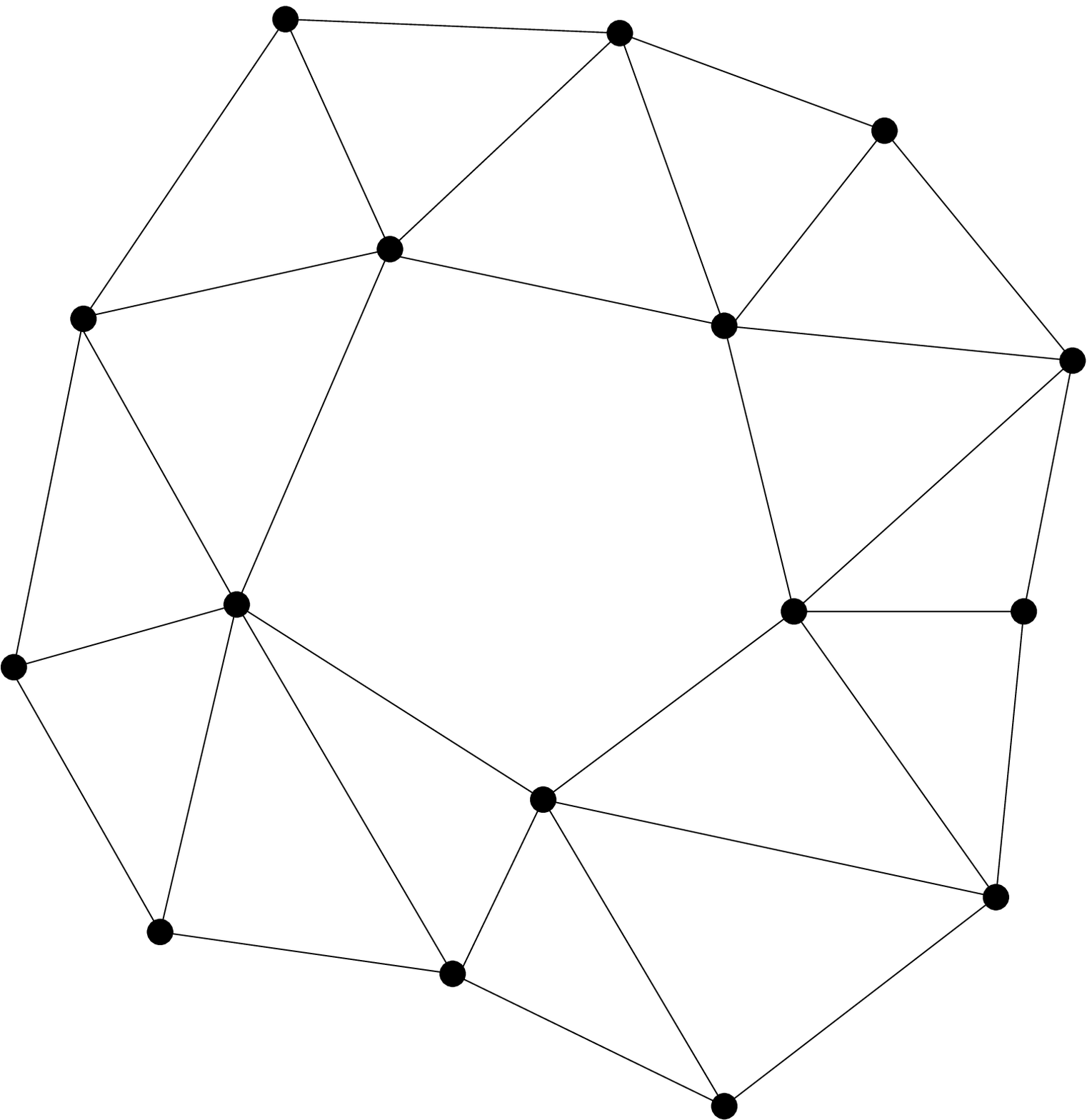}\\
    (b) triangle cycle
    \end{minipage}
    \begin{minipage}[t]{4.1cm}
    \includegraphics[height=.7in]{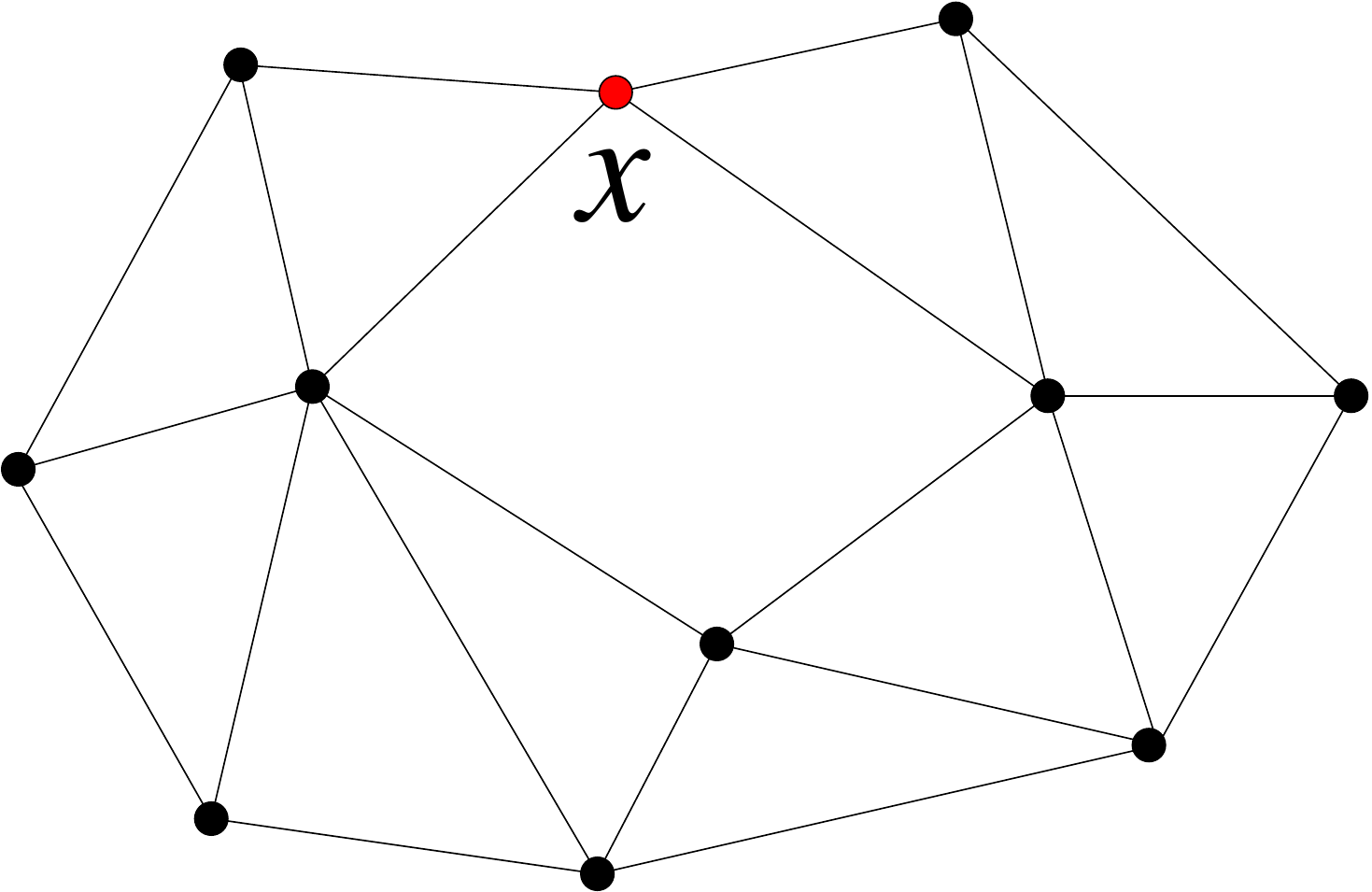}\\
    (c) triangle circuit
    \end{minipage}
    \begin{minipage}[t]{3.4cm}
    \includegraphics[height=.7in]{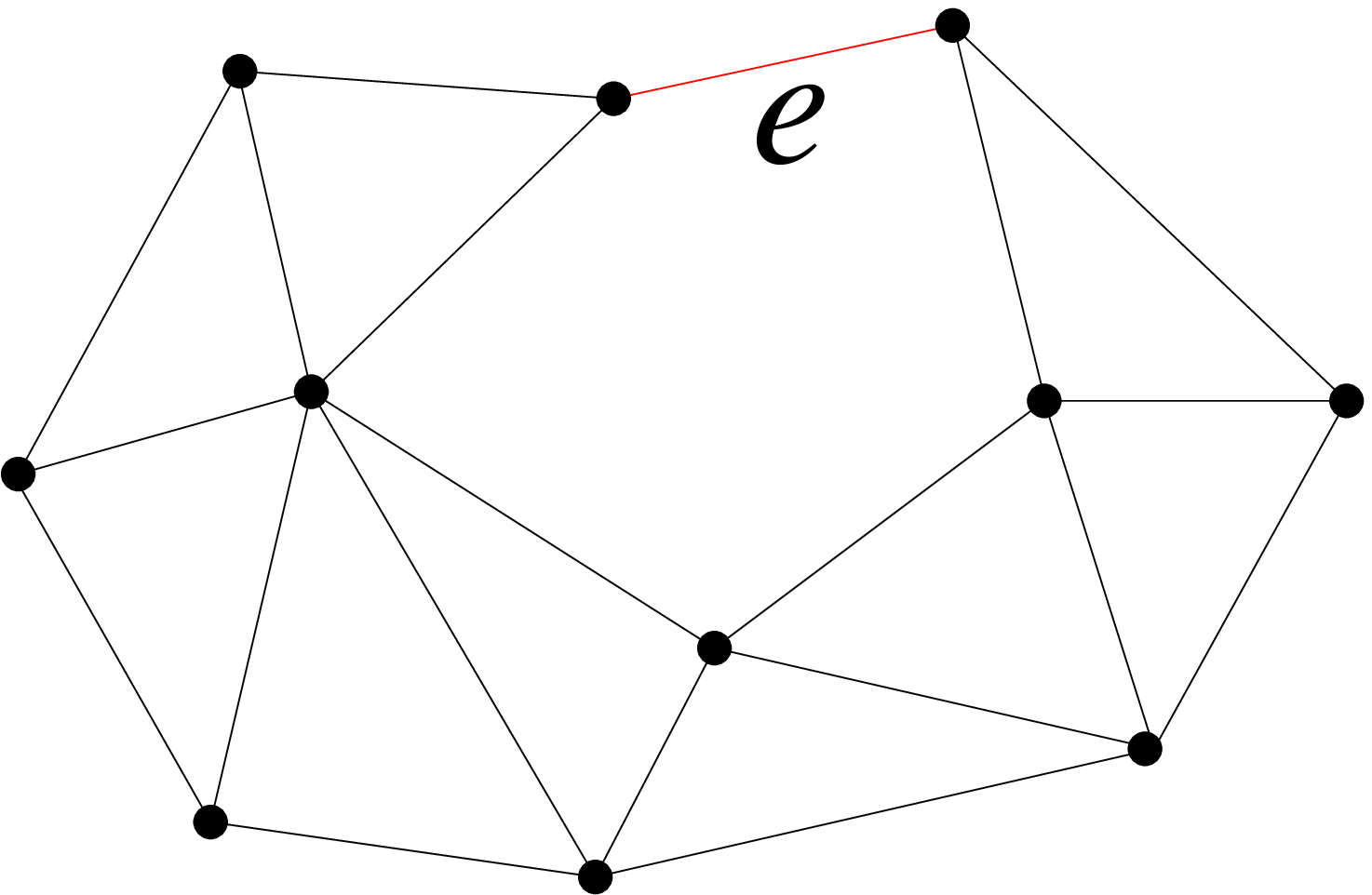}\\
    (d) triangle bridge
    \end{minipage}\\[-1mm]
\caption{Examples of triangle chain, triangle cycle, triangle circuit and triangle bridge}
\label{fig:trianglerChainCycleCircuitBridge}
 \vspace*{-4mm}
\end{figure}
  $G({\cal T})$ is the graph constructed by 
taking the union of the $T_i$s in ${\cal T}$. A node $u$ 
of a triangle $T_i$ is termed a \textit{pendant} of $T_i$, if the edge opposite to $u$ in $T_i$ is 
shared by another triangle in ${\cal T}$. This shared edge is called an \textit{inner side} of 
$T_i$. Each triangle $T_i$ has at least one edge which is not shared by another triangle in ${\cal 
T}$. Such a non-shared edge is called an \textit{outer side} of $T_i$. In 
Fig.~\ref{fig:trianglerChainCycleCircuitBridge}\,(a), $T_4=\{u,v,w\}$ has two pendants $v$ and $w$. 
It has two inner sides $uw$ and $uv$ and one outer side $vw$. If each of $T_1$ and $T_m$ has unique 
and distinct pendants, then $G({\cal T})$ is termed a \textit{triangle chain}. 
Fig.~\ref{fig:trianglerChainCycleCircuitBridge}\,(a) shows an example of triangle chain. By 
construction, a triangle chain involves only flips; hence rigid. If $T_1$ 
and $T_m$ share a common edge other than those shared with $T_2$ and $T_{m-1}$, then the union 
$G({\cal T})$ is called a \textit{triangle cycle}. In a triangle cycle, each triangle has exactly 
two inner and one outer sides. Fig.~\ref{fig:trianglerChainCycleCircuitBridge}\,(b) shows an example 
of a triangle cycle. Every wheel graph is a triangle cycle.

If $G({\cal T})$ is not a triangle cycle and $T_1$ and $T_m$ have a unique pendant in common, then
$G({\cal T})$ is called a \textit{triangle  
circuit}~(Fig.~\ref{fig:trianglerChainCycleCircuitBridge}\,(c)). The common pendant is called a 
\textit{circuit knot}. $x$ is the circuit knot of the triangle circuit.
Let ${\cal T} = (T_1, T_2, \ldots, T_m)$ be  a triangle stream corresponding to a triangle chain.
$T_1$ and $T_m$ have unique and distinct pendants. We connect these pendants by an edge $e$.
$G({\cal T})\cup\{e\}$ is called a \textit{triangle bridge} 
(Fig.~\ref{fig:trianglerChainCycleCircuitBridge}\,(d)). The edge $e$ is called the 
\textit{bridging edge}. The \textit{length of a triangle stream ${\cal T}$} is the number of 
triangles in it and is denoted by $l({\cal T})$.

\begin{lemma} \label{lem:trngleCycleToCircuit} \label{lem:trngleCircuitToBridge}
    1) Every triangle cycle has a \textit{spanning wheel} or \textit{triangle circuit} (a wheel or 
    triangle circuit which is a spanning subgraph of the triangle cycle). 2) Every triangle circuit 
has a \textit{spanning triangle bridge} (a triangle bridge which is     a spanning subgraph of the 
triangle circuit).
\end{lemma}
\begin{proof}
See Appendix~\ref{app:trngleCycleToCircuit} for the first part. For the second part, see 
Appendix~\ref{app:trngleCircuitToBridge}.
\qed \end{proof}

We have seen that a rigid realization in $\mathbb R^2$ may have \textit{flip ambiguity}, 
i.e., it may yield another configuration by applying flip operation only. In $\mathbb
R^2$, if a rigid realization admits no flip ambiguity, then it is globally rigid. Using this,
we shall prove the generically global rigidity as follows.

\begin{lemma} \label{lem:triangleCycleCircuitBridge}
Triangle cycle, circuit and bridge are generically globally rigid.
\end{lemma}

\begin{proof}
A triangle cycle has a spanning wheel or triangle circuit (Lemma~\ref{lem:trngleCycleToCircuit}). A
wheel graph is generically globally rigid. A triangle circuit always has a spanning triangle bridge
(Lemma~\ref{lem:trngleCircuitToBridge}). If we can prove that a triangle bridge is
generically globally rigid, the result will follow.

Let $G({\cal T})$ be a triangle bridge with the triangle stream ${\cal T} = (T_1, T_2, \ldots,
T_n)$ and the bridging edge $e$. Consider a generic configuration of $G({\cal T})$
(Fig.~\ref{fig:trngleBridgeGlblyRgd}). Note that $G({\cal T})$ contains $n+2$ nodes.
\begin{figure}[h]
\begin{minipage}[c]{2in}
    \centering
    \includegraphics[height=.7in]{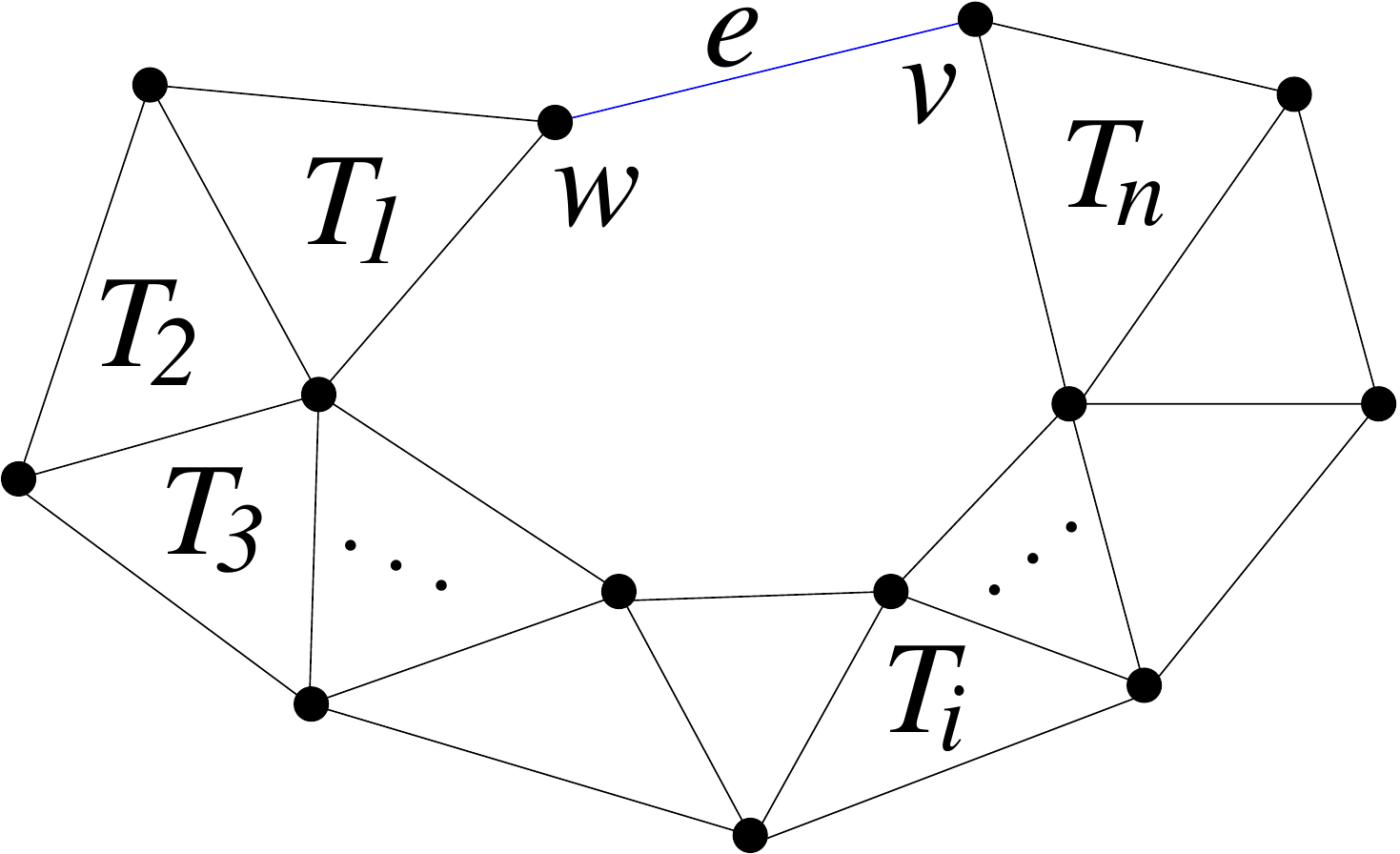}
\end{minipage}
\begin{minipage}[c]{3.5in}
\caption{A generic configuration of a triangle bridge $G({\cal T})$}
\label{fig:trngleBridgeGlblyRgd}
\end{minipage}
\vspace*{-5mm}
\end{figure}
$G({\cal T})-e$ is a \textit{spanning triangle chain} (a triangle chain which is a spanning subgraph 
of the triangle bridge) of $G({\cal T})$. Since $G({\cal T})-e$ can only have flips, i.e., no 
smooth deformation. Therefore, $G({\cal T})$ also admits no smooth deformation.

Consider a different realization of $G({\cal T})$ obtained from the current
realization through a sequence of flip operations. Without loss of generality, we assume that 
$T_1$ remains fixed with positions $(x_0, y_0)$, $(x_1,y_1)$ and $(x_2,y_2)$ for $w$ and
two other nodes in $T_1$ respectively. If $T_2$ is involved in a flip, then the only possibility is
the flip that is taken with respect to the inner edge with $T_1$. Only one point of $T_2$ changes
its position. Let $(x_3,y_3)$ and $(x_3',y_3')$ be the positions of this point in the original and
the modified configuration respectively. From elementary coordinate geometry,
$x_3'$ and $y_3'$ can be expressed in the form of $\frac{\phi_3}{\psi_3}$ and
$\frac{\xi_3}{\eta_3}$ where $\phi_3$, $\psi_3$, $\xi_3$ and $\eta_3$ are non-zero
polynomials of $x_1$, $y_1$, $x_2$, $y_2$, $x_3$ and $y_3$ with integer coefficients such that
$\psi_3\neq 0$ and $\eta_3\neq 0$. Once, the positions of $T_1$ and $T_2$ in the second
configuration are computed (fixed) then only one point of $T_3$ may need to be computed. This node
again may be involved in a flip with respect to the inner edge of $T_2$. If $(x_4,y_4)$ and
$(x_4',y_4')$ are the positions of this point in the original and modified configurations
respectively, $x_4'$ and $y_4'$ can be expressed in the form of $\frac{\phi_4'}{\psi_4'}$ and
$\frac{\xi_4'}{\eta_4'}$ where $\phi_4'$, $\psi_4'$, $\xi_4'$ and $\eta_4'$ are non-zero
polynomials of $x_2$, $y_2$, $x_3'$, $y_3'$, $x_4$ and $y_4$ with integer coefficients such that
$\psi_4'\neq 0$ and $\eta_4'\neq 0$. In this expression, if we substitute $x_3'$ and $y_3'$ by
expressions involving $x_1$, $y_1$, $x_2$, $y_2$, $x_3$ and $y_3$ (obtained from the previous
equations), $x_4'$ and $y_4'$ can be expressed in the form of $\frac{\phi_4}{\psi_4}$ and
$\frac{\xi_4}{\eta_4}$ where $\phi_4$, $\psi_4$, $\xi_4$ and $\eta_4$  are non-zero
polynomials of $x_1$, $y_1$, $x_2$, $y_2$, $x_3$, $y_3$, $x_4$ and $y_4$ with integer coefficients
such that $\psi_4\neq 0$ and $\eta_4\neq 0$. Proceeding in this way, finally $(x_{n+1}',y_{n+1}')$,
the position of $v$, can be expressed in the form of $\frac{\phi}{\psi}$ and $\frac{\xi}{\eta}$
where $\phi$, $\psi$, $\xi$ and $\eta$ are non-zero polynomials of $x_i$s and $y_i$s, $1\leq
i\leq n+1$, the coordinates of the nodes in the first configuration, with integer coefficients such
that $\psi\neq 0$ and $\eta \neq 0$.

Since $v$ and $w$ are adjacent, $d(w,v)$ (the Euclidean distance between $w$ and $v$) remains
preserved in both configurations. In terms of the coordinates,
\begin{eqnarray*}
(x_{n+1}'-x_0)^2+(y_{n+1}'-y_0)^2 & = & (x_{n+1}-x_0)^2+(y_{n+1}-y_0)^2,\\
      \eta^2(\phi-x_0\psi)^2+\psi^2(\xi- y_0\eta)^2 & = & \eta^2\psi^2
                            (x_{n+1}-x_0)^2+ \eta^2\psi^2(y_{n+1}-y_0)^2.
\end{eqnarray*}
So the coordinates in the original configuration  are algebraically dependent.
It contradicts that the configuration is generic. Therefore, no flip is possible.
\qed \end{proof}

\subsection{Triangle notch and  triangle net}
Consider a sequence ${\cal T} =(T_1, T_2, \ldots, T_m)$  of triangles. Suppose, for $i=2$, $3$, 
$\cdots$, $m$, each $T_i$ shares exactly one edge with exactly one $T_j$, $1\leq j < i$. The node 
opposite to this sharing edge is called a \textit{pendant} of $T_i$ in ${\cal T}$. 
Fig.~\ref{fig:triangleTree} shows an example of such a sequence and $x$ is a pendant of $T_2$. $T_1$ 
has no pendant.
\begin{figure}[h]
\begin{minipage}[c]{.19\textwidth}
    \centering
    \includegraphics[height=.7in]{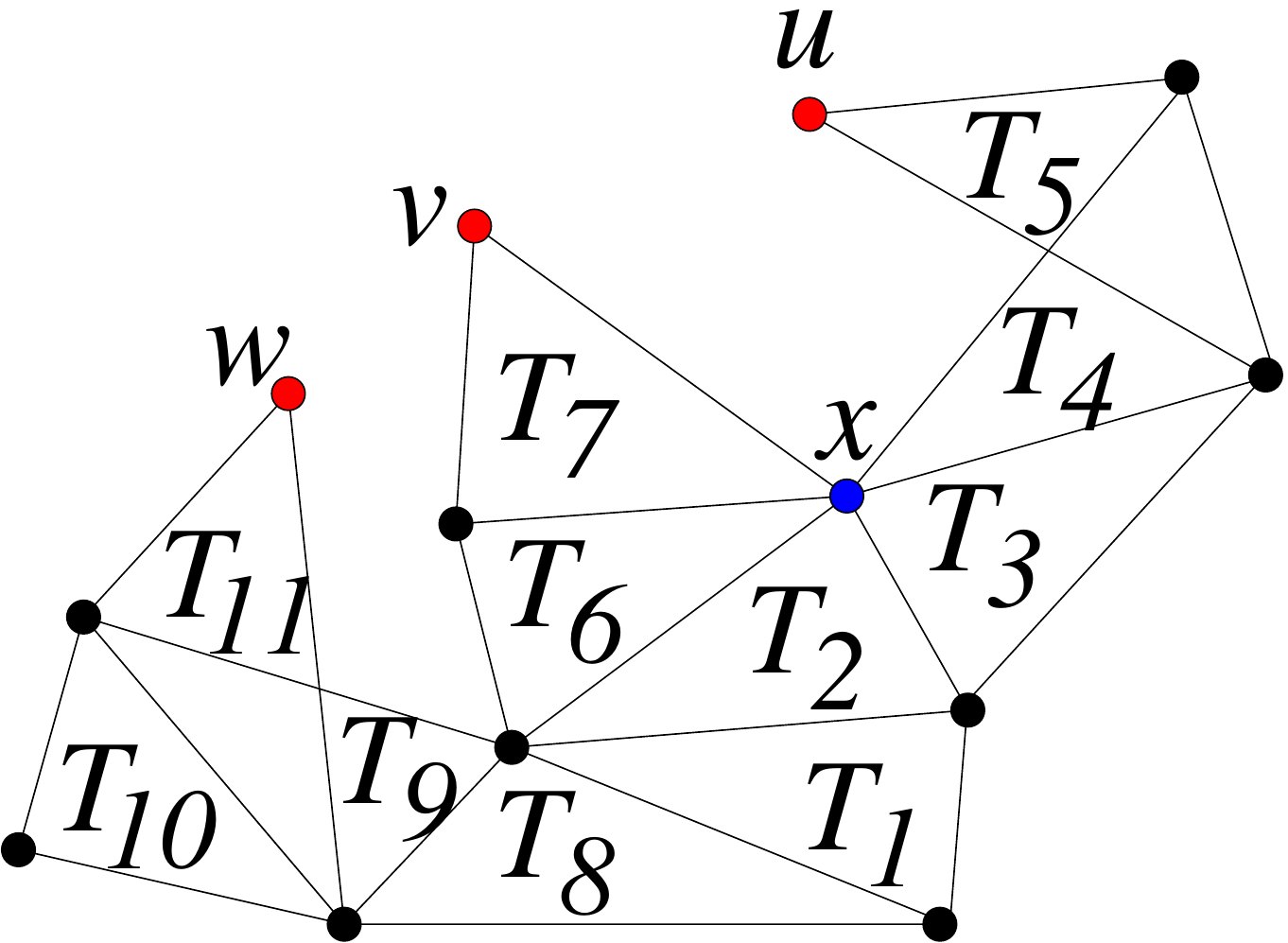}
\end{minipage}
\begin{minipage}[c]{.14\textwidth}
\caption{Triangle tree}
\label{fig:triangleTree}
\end{minipage}
\begin{minipage}[c]{.47\textwidth}
    \flushright
    \includegraphics[height=.8in]{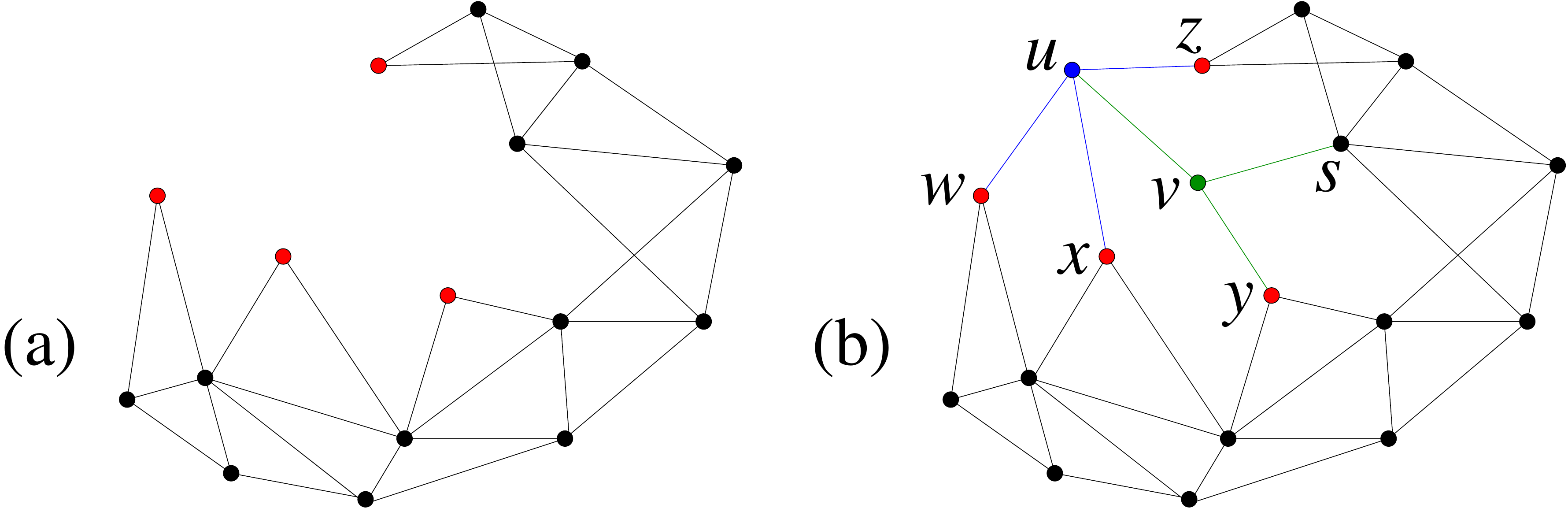} ~~
\end{minipage}
\begin{minipage}[c]{.18\textwidth}
\caption{(a) Triangle tree (b) $u$ and $v$ are Extended nodes}
\label{fig:extendedNode}
\end{minipage}
\end{figure}
For $2\leq i \leq m$,  each $T_i$ has exactly one pendant in $\cal T$. The graph $G({\cal T})$ 
corresponding to such a sequence $\cal T$, is called a \textit{triangle tree}. 
Fig.~\ref{fig:triangleTree} is an example of a triangle tree with $11$ triangles. $G({\cal T})$ 
contains no triangle cycle. Otherwise, there always exists a $T_j$ which shares two edges with some 
triangles before $T_j$ in ${\cal T}$.
If a triangle $T_i$ shares no edge with $T_j$,  $j>i$, is called a \textit{leaf triangle}. A 
leaf triangle shares exactly one edge with other triangles in ${\cal T}$. It has a unique pendant, 
called a \textit{leaf knot}. $T_5$, $T_7$ and $T_{11}$ are leaf triangles and $u$, $v$ and $ w$ are 
leaf knots. By construction, any realization of a triangle tree is rigid.

\begin{definition}
Let $G({\cal T})$ be a triangle tree. A node $v$, outside $G({\cal T})$, is called an
\textit{extended node} of $G({\cal T})$, if $v$ is adjacent to at least three
nodes, each being
$i)$ a pendant in $G({\cal T})$; or $ii)$ an extended node of $G({\cal T})$.
%
Each of the edges which connect the extended node to a pendant or an extended knot of $G({\cal T})$
is called an \textit{extending edge}.
\end{definition}
Fig.~\ref{fig:extendedNode}\,(a) is a triangle tree, say $G({\cal T})$. 
Fig.~\ref{fig:extendedNode}\,(b) consists of a replica of the graph in
Fig.~\ref{fig:extendedNode}\,(a) and some more nodes and edges. Fig.~\ref{fig:extendedNode}\,(a) 
does not contain $u$ of Fig.~\ref{fig:extendedNode}\,(b). $u$ is adjacent to three pendants $w$, 
$x$ and $z$. So $u$ is an extended node of $G({\cal T})$. The edges $uw$, $ux$ and $uz$ are the
extending edges of $u$. Similarly, $v$ is adjacent to an extended node $u$ and two pendants $s$ and 
$y$.  So $v$ is also an extended node of $G({\cal T})$; where $vu$, $vs$ and $vy$ are the extending 
edges.

\begin{definition}
A graph $G$ is called a \textit{triangle notch}, if it can be generated from a
triangle tree $G'({\cal T})$, where $G'$ is proper subgraph of $G$, by adding only one extended node
$v$ where all the leaf knots of $G'({\cal T})$ are adjacent to $v$. The extended node $v$ is called
the \textit{apex} of $G$.
\end{definition}
 Fig.~\ref{fig:triangleNotch}\,(b) shows an example of a triangle notch with the apex $v$. The
triangle tree from which it is generated is separately shown in
Fig.~\ref{fig:triangleNotch}\,(a). 
\begin{figure}[!h]
\begin{minipage}[c]{.6\textwidth}
   {  \centering
    \includegraphics[height=.7in]{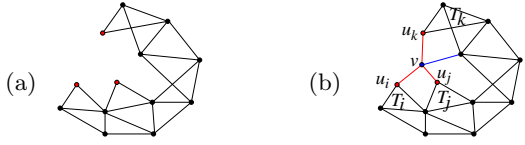}\\[-10mm]
    ~~~~~(a) ~~~~~~~~~~~~~~~~~~~~~~~~~~~~~~~ (b) ~~~~~~~~~}
\end{minipage}
\begin{minipage}[c]{.28\textwidth}
 \caption{(a) Triangle tree $G({\cal T})$ (b) Triangle notch with apex $v$}
\label{fig:triangleNotch}
\end{minipage}
%
\end{figure}

\begin{lemma}\label{lem:triangleNotch}
A triangle notch is generically globally rigid.
\end{lemma}
\begin{proof} 
See Appendix~\ref{app:triangleNotch}.
\qed \end{proof}

\begin{lemma}\label{lem:extendedNode}
Let $G$ be a graph obtained from a triangle tree $G'({\cal T})$ by adding extended nodes, where $G'$
is a proper subgraph of $G$. Any extended node along with all pendants and extended
nodes adjacent to it lie in a generically globally rigid subgraph.
\end{lemma}
\begin{proof}
If the extended node $v$ is adjacent to only pendants of $G'({\cal T})$, then these pendants are
leaf knots of some triangle tree $G''({\cal T'})$ where the triangles of ${\cal T'}$ are all taken
from  ${\cal T}$. $G''({\cal T'})\cup \{v\}$ forms a triangle notch. 
Fig.~\ref{fig:extendedNodeProof}\,(a) shows an example of such a case.
\begin{figure}[h]
\begin{minipage}[c]{.65\textwidth}
\centering
    \includegraphics[height=1in]{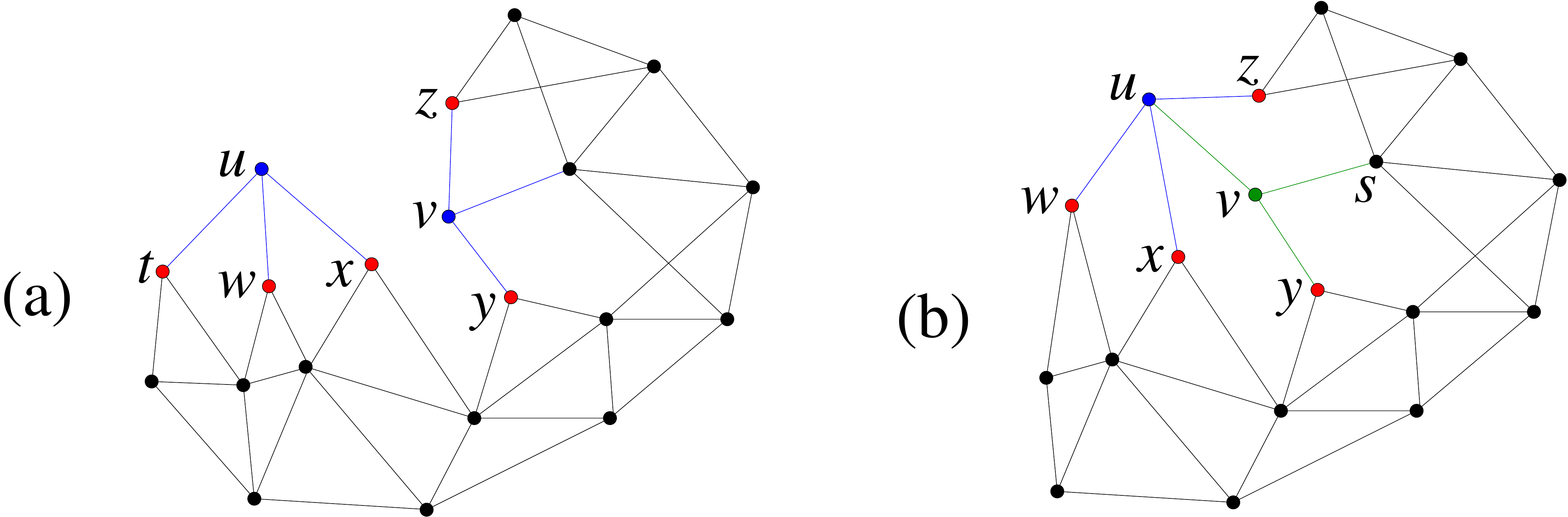}
\end{minipage}
\begin{minipage}[c]{.3\textwidth}
\caption{$u$ and $v$ are extended nodes where (a) $u$, $v$ are adjacent to pendants only,
        (b) $u$, $v$ are adjacent to both pendant and extended nodes} 
\label{fig:extendedNodeProof}
\end{minipage}
\end{figure}
%
 By Lemma~\ref{lem:triangleNotch}, $G''({\cal T'})\cup \{v\}$ is generically globally rigid.
Now consider the case when $v$ is adjacent to at least one extended node. Let $u$ be an extended
node which is adjacent to $v$ (Fig.~\ref{fig:extendedNodeProof}\,(b)). Since $u$ is also an
extended node of $G'({\cal T})$, we assume that $u$ lies in a generically globally rigid subgraph
$G_1$ of $G$ and is generated from a triangle tree $G''({\cal T'})$ by adding extended nodes
(including $u$),  where ${\cal T'}$ contains triangle only from ${\cal T}$. Consider a generic
configuration $\cal P$ of $G_1$. If $G_1$ admits any flip operation in $\cal P$ to
yield a different configuration ${\cal P}'$, then proceeding in a manner similar to that in the
proof of Lemma~\ref{lem:triangleNotch}, we can show that at least three nodes (pendants or extended
nodes adjacent to $v$) are algebraically dependent.
\qed \end{proof}

\begin{definition}
A graph $G$ is called a \textit{triangle net}, if it may be generated from a
triangle tree $G'({\cal T})$ by adding one or more extended
nodes and satisfying the following conditions: 
\begin{enumerate}
 \item $G$ contains no triangle cycle, triangle circuit or triangle bridge; ~ and
 \item there exists an extended node $u$ such that every leaf knot of $G'({\cal T})$ is connected 
to $u$ by a path (\textit{extending path}) containing only       extending edges.
\end{enumerate}
The last extended node added to generate the triangle net is called an \index{apex of triangle
net}\textit{apex} of the triangle net.
\end{definition}

In Fig.~\ref{fig:triangleNet}\,(a),  $u$ and $v$ are two extended nodes. The leaf knots
$t$, $w$ and $x$ are connected to $u$ by extending paths. Other leaf knots $y$ and $z$ are 
connected to $v$ by extending paths. No extending path exists between the $y$ and $u$, and 
$x$ and $v$. So the graph shown in Fig.~\ref{fig:triangleNet}\,(a) is not a triangle 
net. Fig.~\ref{fig:triangleNet}\,(b) contains two extended nodes $u$ and $v$. All the leaf knots 
$w$, $x$, $y$ and $z$ are connected to $u$ by extending paths. 
\begin{figure}[h]
\begin{minipage}[c]{.55\textwidth}
    \hspace{1cm} \includegraphics[height=.9in]{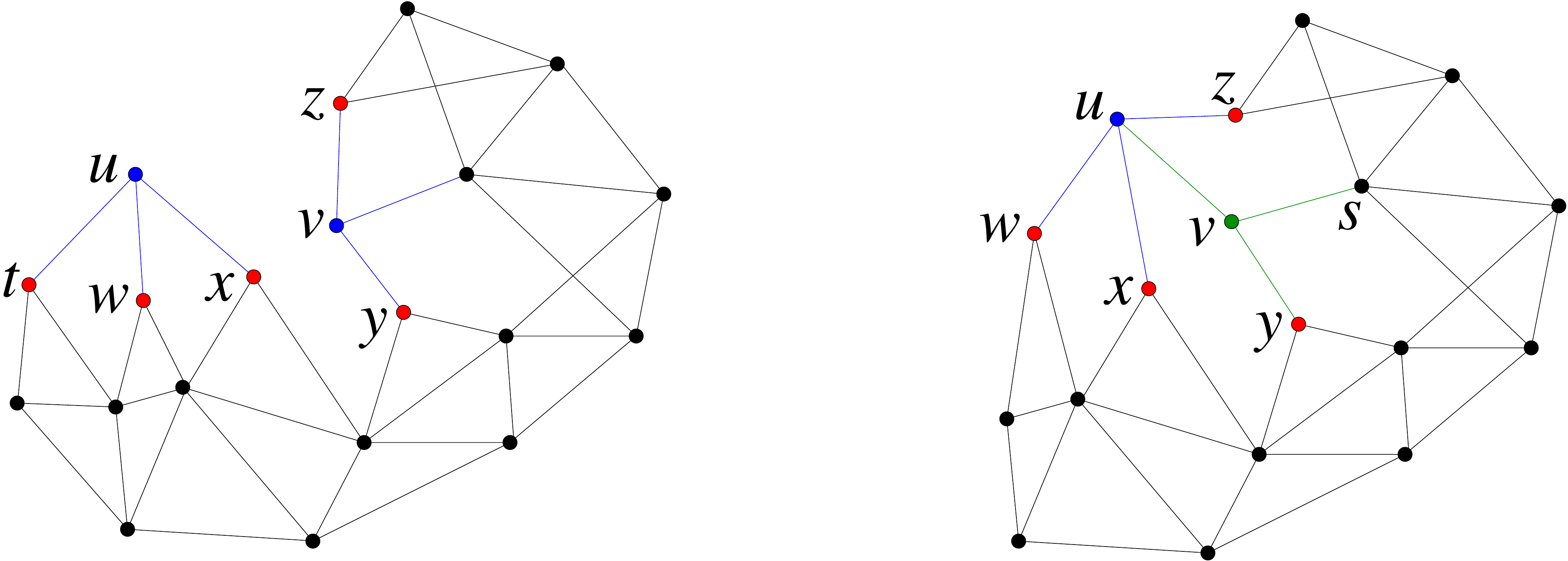}\\[-10mm]
    \hspace*{.3cm}(a) ~~~~~~~~~~~~~~~~~~~~~~~~~~~~~~~~ (b) ~~~~~~~~~
\end{minipage}
\begin{minipage}[c]{.28\textwidth}
\caption{(a) Not a triangle net (b) A triangle net}
\label{fig:triangleNet}
\end{minipage}
\end{figure}
  Thus Fig.~\ref{fig:triangleNet}\,(b)
is an example of a triangle net. The graph shown in Fig.~\ref{fig:triangleNet}\,(b), is generated
from a triangle tree by adding extended nodes $u$ and then $v$. So $v$ is an apex of $G$. Triangle 
notch is a special case of triangle net.

\begin{lemma} \label{lem:triangleNet}
A triangle net is generically globally rigid.
\end{lemma}
\begin{proof}
See Appendix~\ref{app:triangleNet}.
\qed \end{proof}

\subsection{Triangle bar}
A graph $G$ is called a \textit{triangle bar}, if it satisfies one of the followings:
\begin{enumerate}
\item $G$ can be obtained from a triangle cycle, triangle circuit, triangle bridge or triangle
net by adding zero or more edges, but no extra node;

\item $G=B_i\cup B_j$ where $B_i$ and $B_j$ are triangle bars which share at least three nodes;
      or

\item $G=B_i\cup \{v\}$ where $B_i$ is a triangle bar and $v$ is a node not in $B_i$, and adjacent
      to at least three nodes of $B_i$.
\end{enumerate}
Note that triangle cycle, triangle circuit, triangle bridge and triangle net are
also triangle bars. These triangle bars will be referred as \textit{elementary bars}.

Fig.~\ref{fig:triangleBar} shows some examples of triangle bars. The first figure is a triangle
cycle. Next two are triangle nets.
\begin{figure}[h]
\begin{minipage}[c]{.15\textwidth}
\caption{Examples of triangle bar}
\label{fig:triangleBar}
\end{minipage}
\begin{minipage}[c]{.8\textwidth}
\centering
    ~~~\includegraphics[height=.8in]{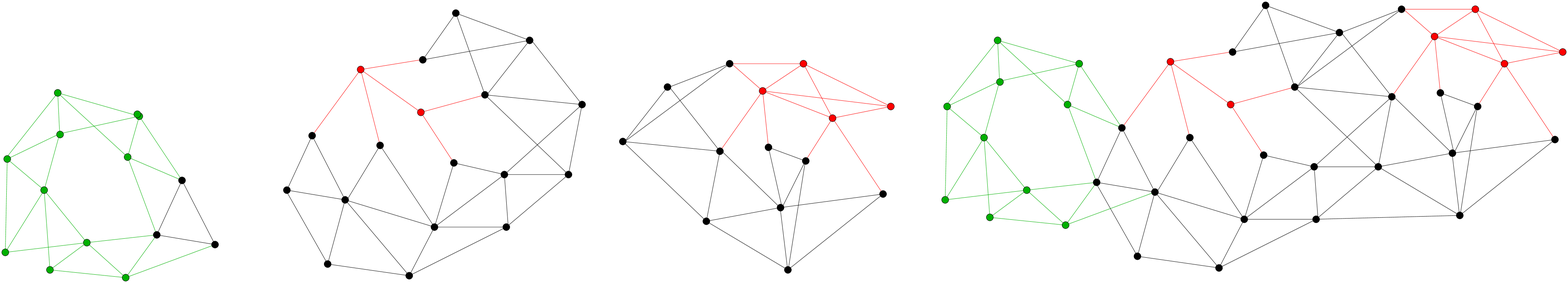}
\end{minipage}
\end{figure}
The fourth figure shows an example of a triangle bar which is obtained by stitching the first three
elementary bars through common triangles.

\begin{theorem} \label{Th:triangleBar}
Triangle bar is generically globally rigid.
\end{theorem}

\begin{proof}
From Lemma~\ref{lem:triangleCycleCircuitBridge} and \ref{lem:triangleNet}, elementary bars are
generically globally rigid. Suppose two triangle bars $B_i$ and $B_j$ share three nodes. Since all
the nodes are in generic position, these three nodes are non-collinear. Using
Lemma~\ref{lem:3commonPoint}, $B_i\cup B_j$ is generically globally rigid.

Let a triangle bar $B$ be obtained from another triangle bar $B'$ by adding a node $v$
which is adjacent to three nodes in $B'$. In a generic realization of $B'$, any
node placed with three given distances from known positions has a unique location. So $B$ is
generically globally rigid.
\qed \end{proof}

\begin{theorem} \label{Th:trilaterionWhlextnInBar}
Trilateration graph and wheel extension are triangle bars.
\end{theorem}

\begin{proof}
See Appendix~\ref{Th:trilaterionWhlextnInBar}.
\qed \end{proof}

\section{Problem statement}
\label{sec:problemAndMapping}

A triangle bar is a class of graphs which includes trilateration and wheel extension graphs as 
special cases. Starting from a triangle of three reference nodes, we find a maximal triangle bar. 
Let $G({\cal T})$ be a triangle tree where ${\cal T} = (T_1$, $T_2$, $\ldots$, $T_n)$. Three nodes 
in $T_1$ are chosen as the reference nodes. This triangle is called \textit{seed triangle}. Our goal 
is to identify a maximal triangle bar containing $T_1$ and then mark the nodes in this triangle bar 
as localizable.

\begin{problem}
Consider a distance graph $G=(V,E,d)$, generically embedded in plane with a seed triangle $T_1$.
Find a maximal triangle bar containing $T_1$ in a distributed
environment and mark the nodes of the triangle bar as localizable.
\end{problem}

We solve the problem involving only connectivity information. No distance information is used. 
Here onwards, we ignore the distance function $d$ and consider the graph $G=(V,E)$. The 
stated problem is solved by exploiting flips of triangles in $G$. In order to solve the problem, we 
introduce the notion of \textit{flip-triangle graph} of $G$.

\subsection{Flip-triangle graph}
Given a graph $G=(V,E)$, we construct a graph ${\cal G} =({\cal V},{\cal E})$ with 
${\cal V}$ = $\{t_1$, $t_2$, $\ldots$, $t_N\}$ where $t_i$ represents a triangle $T_i$ in $G$ and 
$\{t_i,t_j\}\in {\cal E}$ if and only if $T_i$ and $T_j$ share an edge in $G$. The graph ${\cal G}$ 
is termed as flip-triangle graph of $G$, in short $FTG(G)$. If no ambiguity occurs, we use $t_i$ to 
denote a vertex of $FTG(G)$ and $T_i$ to refer the corresponding triangle in $G$. A maximal tree in 
$FTG(G)$ is called a \textit{flip-triangle tree} ($FTT$). A connected $FTG(G)$ has unique $FTT$. 
Let ${\cal T}$ = $(T_1$, $T_2$, $\ldots$, $T_m)$ be a sequence of triangles in $G$ and  $\tau$ 
= $(t_1$, $t_2$, $\ldots$, $t_m)$ be the corresponding sequence of nodes in ${\cal G}$. If no 
ambiguity occurs, ${\cal T}$ also means the subgraph obtained from the union of $T_i$s in 
${\cal T}$. Similarly, $\tau$ means corresponding subgraph in ${\cal G}$. We describe some 
properties which are useful for developing the proposed algorithm.

\begin{proposition} \label{prop:triangleCycleToCycleInFTG}
${\cal T}$ is a triangle cycle of length $n$ in $G$ if and only if $\tau$ is an $n$-cycle in ${\cal 
G}$.
\end{proposition}

\begin{proof}
Let ${\cal T}$ be a triangle cycle in $G$. By construction, pairs of nodes $t_i$ and $t_{i+1}$ for 
$1\leq i \leq n-1$, and $t_m$ and $t_1$ are adjacent in ${\cal G}$. For some $i$, $j$ $(|i-j|>1)$ 
(except $t_1$ and $t_m$), if $t_i$ and $t_j$ are adjacent in ${\cal G}$ then $T_i$ and $T_j$ share 
an edge. This contradicts that ${\cal T}$ is a triangle cycle in $G$. Therefore, $\tau$ is an 
$n$-cycle in ${\cal G}$.

Conversely, let $\tau$ be an $n$-cycle in ${\cal G}$. For some $i$, $j$ $(|i-j|>1)$ (except $T_1$ 
and $T_m$), if $T_i$ and $T_j$  share an edge in $G$ then $t_i$ and $t_j$ are adjacent in ${\cal 
G}$. It contradicts that $\tau$ is an $n$-cycle in ${\cal G}$. Therefore, ${\cal T}$ is a triangle 
cycle in $G$ of length $n$. 
\qed \end{proof}

\begin{proposition} \label{prop:triangleTreeToTreeInFTG}
${\cal T}$ is a triangle tree in $G$ if and only if $\tau$ is a tree in ${\cal G}$.
\end{proposition}

\begin{proof}
Let ${\cal T}$ be a triangle tree in $G$. By construction, $\tau$ is connected subgraph in ${\cal 
G}$. In view of Proposition~\ref{prop:triangleCycleToCycleInFTG}, $\tau$ contains a cycle in ${\cal 
G}$ if and only if ${\cal T}$ contains a triangle tree in $G$. Hence the result follows.
\qed \end{proof}

\begin{proposition} \label{prop:maxTriangleTreeToFTT}
${\cal T}$ is a maximal triangle tree in $G$ if and only if $\tau$ is an $FTT$ in $FTG(G)$.
\end{proposition}

\begin{proof}
Let ${\cal T}$ be a maximal triangle tree in $G$. $\tau$ is a tree in ${\cal 
G}$~(Proposition~\ref{prop:triangleTreeToTreeInFTG}). If $\tau$ is not an $FTT$ in $FTG(G)$, 
there exists a tree $\tau'$ containing $\tau$ as a proper subtree in ${\cal G}$. The 
triangle tree ${\cal T}'$ corresponding to $\tau'$ also contains ${\cal T}$ as a proper subgraph in 
$G$. This is contradicts that ${\cal T}$ is a maximal triangle tree in $G$.

Conversely, let $\tau$ is an $FTT$. If ${\cal T}$ is not maximal triangle tree in $G$, there is a 
${\cal T}'$ containing ${\cal T}$ as a proper subgraph. ${\cal T}'$ corresponds a 
tree $\tau'$ in ${\cal G}$~(Proposition~\ref{prop:triangleTreeToTreeInFTG}) while $\tau'$ contains 
$\tau$ as a proper subtree. This is a contradiction. Hence the result follows.
\qed \end{proof}

\subsection{Solution plan}
Consider a graph $G=(V,E)$. A triangle bar may be identified in $G$ by three rules as in its 
definition. First, we find elementary bars in $G$. If possible, then we stitch them via three common 
nodes to form a larger triangle bar; or extend a triangle bar ${\cal B}$ successively by adding a 
new node which is adjacent to at least three nodes of ${\cal B}$. 
After computing the $FTG(G)$, the stated problem is solved in a distributed set up as follows:

\begin{enumerate}
\item We identify all the components of ${\cal G}$. For each component ${\cal G}'$ in ${\cal G}$, we 
compute a corresponding spanning tree $FTT({\cal G}')$ which is a maximal subtree in ${\cal  G}$.

\item Finding triangle cycles is equivalent to finding the cycles in ${\cal G} = 
FTG(G)$~(Proposition~\ref{prop:triangleCycleToCycleInFTG}). We identify a \textit{set of base 
cycles} (a minimal set of cycles such that any cycle of the 
graph may be obtained by union of some base cycles and deleting some parts). 

 \item A triangle chain is also a triangle tree. The generator chains of triangle circuits 
and bridges and generator trees of triangle nets are uniquely identified by subtrees in 
$FTG(G)$~(Proposition~\ref{prop:triangleTreeToTreeInFTG}). We identify other elementary bars in $G$ 
from the $FTT$s by suitable extensions.

\item Finally, we stitch or extend these elementary bars to form a maximal triangle bar in $G$ 
containing $T_1$; then we mark the nodes in this triangle bar as localizable.
\end{enumerate}

\section{Localizability testing}
\label{sec:localizabilityTesting}
This section describes a distributed technique to find the maximal triangle bar with a seed 
triangle $T_1$  in three phases. This triangle bar is reported as the localizable subgraph of 
$G$.\vspace*{-4mm}

\subsection{Representation of graph and flip-triangle graph}
Each node contains data structures suitable for describing and storing necessary information for the
execution of the algorithm. We assume that each node contains a unique number as its identification
(called \textit{node-id}) and a list ($\mt{nbrs}$) of node-ids of its neighbours. The node contains 
no edge distance information. To represent a triangle in computer, we define a data structure, with 
type name $\mt{\mathbf{Trngl}}$, containing: ~ 1)~node-ids of the nodes of the triangle; and ~~2)~a 
list of adjacent nodes in $FTG(G)$ (i.e., triangles sharing its edge in $G$).
In a distributed environment, the node with minimum node-id among three nodes of a triangle is
designated as the  \textit{leader}. The leader contains all the information of the triangle
and processes them. Each node $v$ additionally contains a list ($\mt{trngls}$) of all the triangles
containing $v$ as the leader.\vspace*{-3mm}

\subsection{Communication protocols for $G$ and $FTG(G)$}
A communication between two adjacent nodes in $FTG(G)$ (i.e., two triangles sharing an edge in $G$)
means communication between their leaders which may involve at most 2-hop communication in $G$.
Intermediate communications via other nodes uses standard communication tools for $G$ (i.e.,
communication within $G$). By a communication between a node $s_i\in G$ and a node $t_j\in FTG(G)$
(i.e., triangle $T_j$), we mean the communication between $s_i$ and the leader of $T_j$. 
The nodes of $G$ and $FTG(G)$ use different types of signals to indicate the types of the contents
of the messages. We list these signals as follows:\vspace*{-5mm}
\begin{table}[h!]
\centering
\begin{tabular}{|l|l|}
\hline
Signal types       &   ~~~~~~~~~~~~~~~ Significance of the symbol \\ \hline\hline
\textbf{visit}     &   On arrival of this signal a node of $FTG(G)$ wake up and starts
                                processing. \\ \hline
\textbf{visitNode} &   On arrival of this signal a node of $G$ wake up and starts
                                processing. \\ \hline
\textbf{child}     &   A node in $G$ sends a \textbf{\texttt{child}} signal to its parent to
                                register itself as a child in parent. \\ \hline
\textbf{cycle}     &   A node in $G$ or $FTG(G)$ sends a \textbf{\texttt{cycle}} signal
                       on identification of an elementary bar. \\ \hline
\end{tabular}
\end{table}

\subsection{Phase I: Computing the $FTG$}
Phase-I of the algorithm sets up the basic structure of the $FTG(G)$ using the procedures
\Call{recvNbrList}{\,} and \Call{recvTriangle}{\,} described with pseudo-codes as follows.
\begin{algorithm}[h!]
\begin{algorithmic}[1]
\Procedure{recvNbrList}{\,}  \hfill  \Comment{$s_i$ = current node}
    \cmd{Wake}{on arrival of neighbours ($\mt{nbrs_j}$) form $s_j$}
    \For{(each common $s_k$ in $\mt{nbrs_j}$ and $\mt{nbrs_i} $)}%
        \hfill \Comment{A triangle $T(\triangle s_is_js_k)$ is identified.}
        \If{($s_i$ is the leader of $T$)}  \hfill \Comment{$T$ corresponds new node in $FTG(G)$}
            \For{(each $T'\in s_i.\mt{trngls}$ sharing an edge with $T$)} 
                \hfill \Comment{Found a new edge $TT'$}
                \cmd{Push}{$T'$ into $T.\mt{trngls}$ and $T$ into $T'.\mt{trngls}$ as a
                    neighbour of each other in ${\cal G}$.}
                \hcomment{18mm}{Both of these push operations occur in 
                    $s_i$, since $s_i$ is the leader $T$ and $T'$.}
            \EndFor
            \cmd{Store}{the new triangle $T$ into $s_i.\mt{trngls}$ in its leader.}
            \cmd{Send}{the triangle $T$ to $s_j$ and $s_k$ to find and set edges with other 
                triangles.}
        \EndIf
    \EndFor
\EndProcedure
\end{algorithmic}
\hrule
\begin{algorithmic}[1]
\Procedure{recvTriangle}{\,}\hfill  \Comment{$s_i$ = current node}
    \cmd{Wake}{on arrival of a triangle $T(\triangle s_js_ks_l)$ such that $j<k<l$}
    \For{(each $T'\in s_i.\mt{trngls}$ sharing an edge with $T$)}
        \hfill \Comment{Found a new edge $TT'$ for ${\cal G}$}
            \cmd{Push}{$T$ into $T'.\mt{trngls}$ as a neighbour of $T'$ in $FTG(G)$.}
            \hfill \Comment{$s_i$ is the leader of $T'$}
    \EndFor
    \iIf{($i\in \{j,k,l\}$, say $i=l$)}
        \hfill \Comment{Note that beyond $1$-hop, $i\notin \{j,k,l\}$}
        \cmd{Send}{the triangle $T$ to all neighbours of the current node except $s_j$ and $s_k$}
\EndProcedure
\end{algorithmic} 
\end{algorithm}

\begin{proposition}  \label{prop:createFTG} \label{prop:syncPhaseI} \label{prop:progressPhaseI}
\begin{enumerate}
\item All the processes in Phase I are synchronized. 
\item The algorithm guarantees the progress and finite termination of Phase I.
\item Number of communications from each node is thrice the number of neighbours.
\item Phase I of the algorithm computes the $FTG(G)$.
\end{enumerate}
\end{proposition}

\begin{proof}
On arrival of a triangle message $T$, \Call{recvTriangle}{\,} wakes up. If the node is the leader of 
$T$, it stores the triangle information. It also checks, if $T$ shares an 
edge with the existing triangles in the list $\mt{trngls}$. Note that each triangle in the list 
$\mt{trngls}$ share an edge with $T$. We call these triangles as \textit{neighbour triangle}s of 
$T$. 

\textit{Synchonization}: Each node starts by sending its neighbour list ($\mt{nbrs}$) to every 
neighbour.  When a node $s_i$ receives the neighbour list from a node $s_j$, it identifies 
all triangles which contain $s_i$ and $s_j$ as two nodes. Each of two processes in Phase I is 
atomic. Any order of insertions of triangles into the neighbour triangle list will finally give the 
same result.

\textit{Progress and finite termination}: Every node executes \Call{recvNbrList}{\,} exactly one 
for a neighbour list each neighbour. Each \textbf{For} loop runs over neighbour lists which are 
finite in size. Since the processes are atomic and loops runs on finite neighbour list, progress is 
guaranteed. In \Call{recvTriangle}{\,}, \textbf{If} block conditions are false beyond $1$-hop from 
$s_i$ and stops resending $T$. Assumed channels are reliable, every message reaches its destination 
in finite time. We use Lamport's logical clock. In each node, the value of logical clock does 
not exceed thrice the number of neighbours; one neighbour list and two triangle messages from 
each neighbour.  It also follows the number of communications from a node.

\textit{Computation of $FTG(G)$}: Each node receives a neighbour list of a neighbour exactly once. 
Consider an arbitrary triangle $T(\triangle s_is_js_k)$ in $G$ (Fig.~\ref{fig:phaseI}) while $i < j 
<k$. Let $T$ is received by $s_i$, the leader node of $T$. \Call{recvNbrList}{\,} inserts 
 $T$ into $s_i.\mt{trngls}$. No other node incorporates $T$, though $s_j$ and $s_k$ also identify 
$T$. It may be adjacent to the edges in $FTG(G)$ due to sharing its edges with other triangles. 
\begin{figure}[h]
\begin{minipage}[c]{.4\textwidth}
\flushright
    \includegraphics[height=.9in]{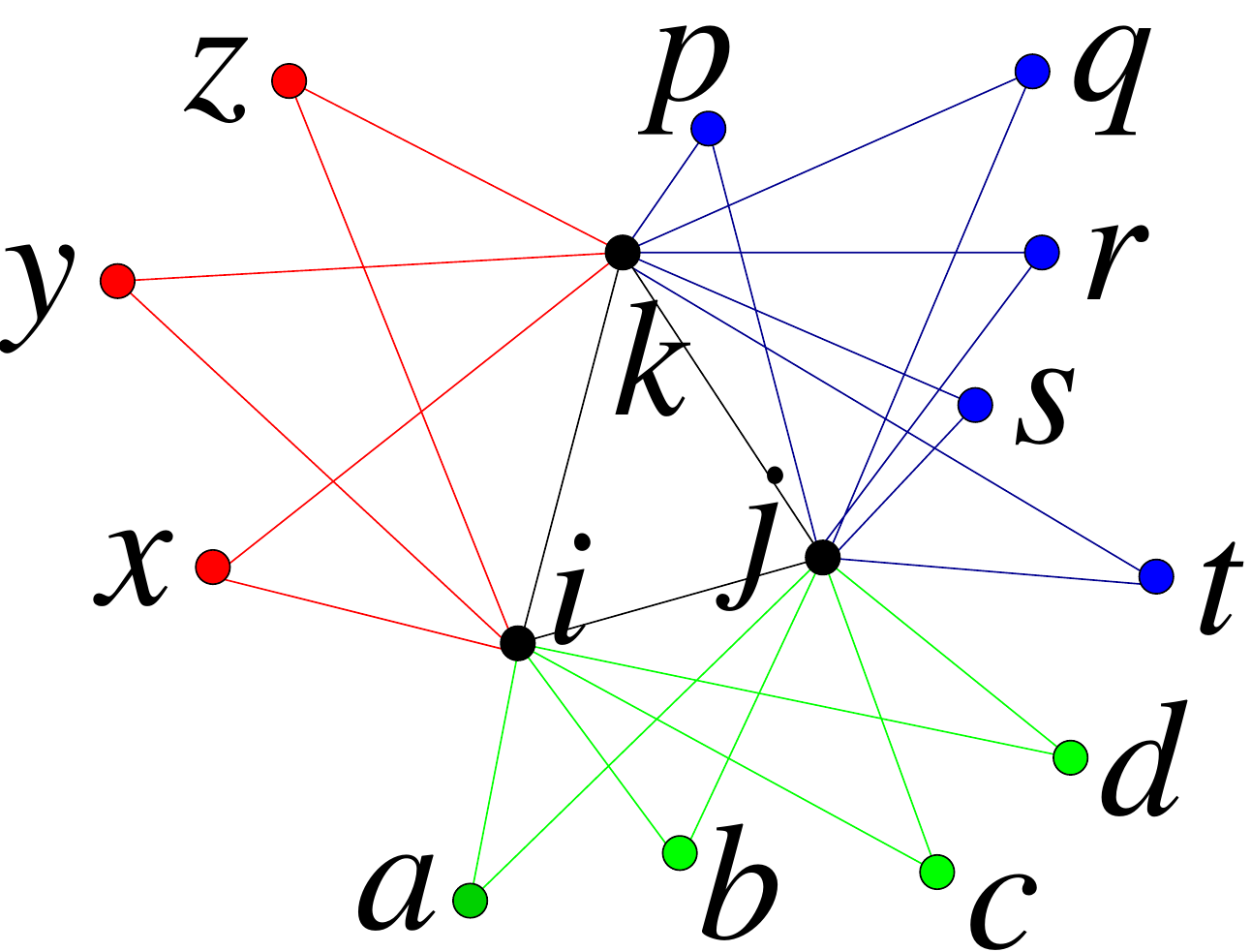} ~~~~
\end{minipage}
\begin{minipage}[c]{.45\textwidth}
\caption{Examples showing all possible triangle communications}
\label{fig:phaseI}
\end{minipage}
\end{figure}
The inner \textbf{for} loop in \Call{recvNbrList}{\,} sets up the edges which are obtained while 
$T$ shares an edge with the triangles whose leader is $s_i$. $s_i$ also sends a triangle 
message $T$ to $s_j$ and $s_k$ to find and set other edges with $T$ in $FTG(G)$. $T$ shares no 
edge in $G$ beyond $1$-hop. The edges in $FTG(G)$ between $T$ and other triangles with leaders other 
than $s_i$ are set by  \Call{recvTriangle}{\,}. When $s_j$ (or $s_k$) receives $T$, it 
checks whether $T$ shares any triangle in its list $\mt{trngls}$ sends to sends $T$ to its 
neighbours other than $s_i$ and $s_j$. Thus $T$ reaches all the nodes which contains a 
triangle that may share an edge with $T$.
\qed \end{proof}

\subsection{Phase II: Finding $FTT$s and elementary bars}
For finding elementary bars, we may assume that $G$ is connected; otherwise, we proceed with the 
component containing the seed triangle $T_1$. Phase-II identifies the $FTT$ of ${\cal G}$ containing 
$t_1$(i.e., $T_1$ in $G$); then elementary bars. This process is triggered by $t_1$ in ${\cal G}$  
sending a 
$\mathbf{visit}$ signal to its adjacent nodes in ${\cal G}$.

\subsubsection{Representation of $FTT$s of $G$ and  elementary bars:}
For this purpose, each node in ${\cal G}$ (a triangle in $G$) contains additional five fields:
\begin{enumerate}
\item $\mt{status}$ (assumes \textbf{0} or \textbf{visited}, initially \textbf{0}) is used to
      indicate the status regarding the processing of the node. After the required processing of 
      a node, $\mt{status}$ is set to \textbf{visited}.

\item $\mt{elementLst}$ is a list of elementary bars in which this particular node (triangle) is
      a constituent part; initially it is empty.

\item $\mt{parent}$ contains the immediate ancestor of the triangle in a $FTT$ of $G$. The seed
      triangle has no parent; and it is treated as root node.

\item $\mt{children}$ is a list containing direct descendants in ${\cal G}$. These are used
to set the trees in ${\cal G}$.

\item $\mt{hearLst}$ holds the list of received triangles with some extra information that help in
      identifying elementary bars.
\end{enumerate}
Each node in $G$ also contains similar five fields to store the information regarding pendants,
extended knots and elementary bars.

\subsubsection{Finding the $FTT$ and base cycles in ${\cal G}'$:}

The $FTT$ of ${\cal G}$ containing $T_1$ are identified by distributed $BFS$ on ${\cal G}$. If we 
take any $FTT$ and add to it a new edge from ${\cal E}$ a set of base cycles of ${\cal G}$ may be 
obtained. The details of these steps are described in procedure \Call{visitTriangle}{\,}.

On arrival of $\mathbf{visit}$ signal into $t_i\in {\cal V}$, \Call{visitTriangle}{\,} in leader
node of $T_i$ sets its status to $\mathbf{visited}$.  $t_i.\mt{parent}$ is assigned the value $t_j$.
This helps in backtracking in the tree in $FTG(G)$. After visiting $t_i$, the process sends
$\mathbf{visit}$ signal to all neighbours in $FTG(G)$. At the same time, $t_i$ sends a
$\mathbf{visitNode}$ signal with $T_i$ to the pendant $s_k$ of $T_i$ with respect to $T_j$ in $G$
for finding elementary bars other than triangle cycles.
\begin{algorithm}[h!]
\begin{minipage}[t]{.48\textwidth}
\begin{algorithmic}[1]
\Procedure{visitTriangle}{\,} \hfill
\Comment{$t_i$= current node of ${\cal G}'$}
    \cmd{Wake}{on a $\mathbf{visit}$ signal from $t_j\in {\cal V}$}
    \If{($\mt{status} \neq \mathbf{visited}$)}
        \cmd{Set}{$\mt{parent}\gets t_j$, $\mt{status} \gets \mathbf{visited}$}
        \cmdp{10mm}{Send}{a $\mathbf{visit}$ signal to all adjacent nodes in ${\cal G}'$
            other than $t_j$}
        \cmdp{10mm}{Send}{a $\mathbf{child}$ signal to $t_j$ to indicate that
            $s_i$ is a child of $s_j$}
        \cmdp{10mm}{Send}{a $\mathbf{visitNode}$ signal to $s_k\in T_i-T_j$ with
              $T_i$ qualified as $\mathbf{triangle}$ }
    \Else \hfill \Comment{a base cycle is identified by $t_it_j$}
        \cmdp{10mm}{Push}{$T_iT_j$ into $\mt{elementLst}$ to indicate that $t_i$
                   lies in the cycle $T_iT_j$}
        \cmdp{10mm}{Send}{a $\mathbf{cycle}$ signal to $t_j$ and $\mt{parent}$
                   with $T_iT_j$ as elementary bar-id}
    \EndIf
\EndProcedure
\end{algorithmic}
\hrule ~

\begin{algorithmic}[1]
\Procedure{elementaryBars}{$x_i$} \hfill \Comment{$x_i$ is the current ($t_i\in
{\cal V}$ or $s_i\in V$)}
    \cmdp{5mm}{Wake}{on arrival of a $\mathbf{cycle}$ signal with elementary
        bar identity as $bar$-$id$}
    \If{($bar$-$id$ is qualified as $\mathbf{delete}$)}
        \cmd{Delete}{$bar$-$id$ from $\mt{elementLst}$}
    \ElsIf{($bar$-$id \notin \mt{elementLst}$)}
        \cmd{Push}{$bar$-$id$ into $\mt{elementLst}$}
        \cmd{Send}{a $\mathbf{cycle}$ signal to $\mt{parent}$
                   with $bar$-$id$}
    \Else \hfill
        \Comment{Elementary bars in ${\cal G}'$ are traced}
        \cmdp{10mm}{Send}{a $\mathbf{cycle}$ signal with $bar$-$id$
             qualified as $\mathbf{delete}$ to its $\mt{parent}$}
    \EndIf
\EndProcedure
\end{algorithmic}
\end{minipage}
\begin{minipage}[t]{.51\textwidth}
\begin{algorithmic}[1]
\Procedure{visitNode}{\,}
\hfill\Comment{$s_i$=current node in $G$}
    \cmd{Wake}{on a $\mathbf{visitNode}$ signal with $X_j$
                                ($=T_j$ or $s_j$)}
    \cmd{Push}{$X_j$ into $\mt{hearLst}$}
    \If{($\mt{status}\neq\mathbf{visited}$)}
        \cmd{Set}{$\mt{parent}\gets X_j$, ~ $\mt{status}\gets\mathbf{visited}$}
        \iIf{($X_j = T_j$)}\hfill\Comment{$s_i\in T_j$, let $T_j=\triangle s_is_ks_m$}
            \cmdp{15mm}{Send}{a $\mathbf{visitNode}$ signal with $s_i$ to all
                $s_m\in \mt{nbrs}-\{s_k,s_l\}$}
        \cmdp{10mm}{Send}{a $\mathbf{child}$ signal to $\mt{parent}$ to inform
            that $s_i$ is a child}
    \ElsIf{($X_j=T_j$)}\hfill\Comment{if $\mt{parent}=T_k$}
        \cmdp{10mm}{Send}{a $\mathbf{cycle}$ signal to $t_j$ and $t_k$ with
            $T_js_iT_k$ as elementary bar-id}
    \ElsIf{($X_j=s_j$)}
        \If{($\mt{parent}=T_k$)} \hfill \Comment{a bridge or net}
            \cmdp{15mm}{Send}{a $\mathbf{cycle}$ signal to $s_j$ and $t_k$ with
                $s_js_iT_k$ as elementary bar-id}
        \ElsIf{($\mt{parent}=s_k$)}
            \If{($\mt{mark}=\textbf{extended}$)}
                \cmdp{20mm}{Send}{a $\mathbf{visitNode}$ signal with $s_i$ to
                    all $s_m\in \mt{nbrs}-\{s_k\}$}
                \cmdp{20mm}{Send}{a $\mathbf{cycle}$ signal to $s_j$ and $s_k$
                    with $s_js_is_k$ as elementary bar-id}
            \ElsIf{($size(\mt{hearLst}) = 3$)}
                \cmd{Set}{$\mt{mark} \gets \mathbf{extended}$}
                \ForAll{($s_j\in\mt{hearLst}$)}
                    \cmdp{25mm}{Send}{a $\mathbf{cycle}$ signal to $s_k$ and $s_j$
                        with $s_js_is_k$ as bar-id}
                \EndFor
            \EndIf
        \EndIf
    \EndIf
\EndProcedure
\end{algorithmic}
\end{minipage}
\end{algorithm}
It sends back a
$\mathbf{child}$ signal to the sender $t_j$ to inform itself as a child. The $\mathbf{child}$ signal
is handled (handler is not described separately) by inserting its sender into the list
$\mt{children}$. If $t_i$ is already visited by some other node $t_k$, a base cycle in $FTG(G)$ is
identified with $t_it_j$ and sends a $\mathbf{cycle}$ signal with $t_it_j$ to $t_j$ and $t_k$. A
cycle in $FTG(G)$ corresponds a triangle cycle in $G$
(Proposition~\ref{prop:triangleCycleToCycleInFTG}). A triangle cycle $G$, corresponding to a base
cycle in $FTG(G)$, contains at least one new triangle which is not a part of any other triangle
cycle in $G$. The outcome of the procedure may be summarized below in an proposition.

\begin{proposition} \label{prop:allTriangleCyces}
\Call{visitTriangle}{\,} identifies the $FTT$ and triangle cycles of $G$ containing $T_1$.
\end{proposition}

\subsubsection{Finding triangle circuits and triangle bridges in ${\cal G}'$:}
Triangle cycles are identified by handling $\mathbf{cycle}$ signals. In view of
the Proposition~\ref{prop:maxTriangleTreeToFTT}, a $FTT$ generate maximal
triangle tree in $G$. These maximal trees provide maximal triangle bars
including appropriate edges and extended knots. The steps are described in the
procedure \Call{visitNode}{\,}.

On arrival of a $\mathbf{visitNode}$ signal with $X_j$ (either a triangle $T_j$
or a node $s_j$), a node $s_i\in V$ wakes up and \Call{visitNode}{\,} stores
$X_j$s into a list $\mt{hearLst}$. If $s_i$ is being visited for the first time
as a node in $G$, $s_i$ set its $\mt{status}$ as \textbf{visited} and
$\mt{parent}$ to $X_j$ for backtracking. If $X_j=T_j$ and $s_i$ is a pendant of
$T_j$, $s_i$ sends a $\mathbf{visitNode}$ signal with $s_i$ to all neighbours
other those in $T_j$. Otherwise, if $s_i$ 1) receives a triangle $T_j$ and its
parent is a triangle $T_k$, a triangle circuit is identified; 2) receives a node
$s_j$ and its parent is a triangle $T_k$ (either a pendant or extended knot), a
triangle bridge or net is identified; and 3) if three $\mathbf{visitNode}$
signals, it identifies itself as an extended knot and as well as a triangle net.
If $s_i$ marks himself as an extended knot it sends $\mathbf{visitNode}$ signal
with $s_i$ to all neighbours except its parent for identifying other extended
knots and nets. Note that $\mathbf{visitNode}$ signal with triangle is sent only
to a pendant from the process \Call{visitTriangle}{\,}. The final outcome of
this process is described below.

\begin{proposition} \label{prop:allCircuitBrdgeNet}
All triangle circuits and triangle bridges in ${\cal G}'$ are identified by \Call{visitNode}{\,}.
\end{proposition}

\subsubsection{Identifying triangle nets in ${\cal G}'$:}

On arrival of a $\mathbf{cycle}$ signal into $x_i$ (either a node in $FTG(G)$ or
pendant or extended knots in $G$), \Call{elementaryBars}{\,} inserts elementary
bar-id into the $\mt{elementLst}$ and in turn sends the same $\mathbf{cycle}$
signal to its parent until it finds a node in $FTG(G)$ containing same
elementary bar-id in respective $\mt{elementLst}$. 
If a matching bar-id is found, it sends $\mathbf{cycle}$ signal with this bar-id
qualified as $\mathbf{delete}$. If it finds bar-id qualified as
$\mathbf{delete}$ and bar-id is deleted from $\mt{elementLst}$ and sends
$\mathbf{delete}$ signal same bar-id qualified as $\mathbf{delete}$ to its
parent until the root node $t_1$ in $FTG(G)$ is reached.

\subsection{Phase III: Identifying the maximal triangle bar containing $T_1$}

\Call{visitNode}{\,} and \Call{elementaryBars}{\,} provide a maximal triangle bar for the $FTT$. 
Thus, we have obtained a set ${\cal M}$ of maximal triangle bars for the $FTT$ containing $T_1$. 
Maximal triangle bars for a $FTT$ do not share any triangle in this $FTT$. From ${\cal M}$, 
two maximal triangle bars $M_1$ and $M_2$ which have three nodes in common are replaced  by 
$M_1\cup M_2$. Repeat these until no replacement. This task may be
achieved by sending a special signal from all nodes of $M_1$ (assuming that
$M_1$ contains $T_1$) to their neighbours. Another $M_i$ sends
this special signal if it hears this signal from three nodes. Finally, a
maximal triangle bar in $G$ will be identified.

\begin{proposition} \label{prop:timeForMaxBarG}
The maximal triangle bar of $G$ containing $T_1$ is identified in polynomial time with $O(|E|)$
one-hop communications over the network.
\end{proposition}

\paragraph{Mark localizable nodes:}
At the end of Phase-II when
$t_1.\mt{elementLst}$ is empty, $t_1$ set its $\mt{status}$ as
$\mathbf{localizable}$. The triangle $T_1$ (i.e., $t_1$) triggers the Phase-III
by sending the $\mt{elementLst}$ to all its adjacent nodes ($\mt{children}$) in
${\cal G}$. On arrival of an elementary bar list into $x_i$ ($t_i\in FTG(G)$ or
$s_i\in G$), \Call{markLocalizable}{\,} starts execution.
\begin{algorithm}[h]
\begin{algorithmic}[1]
\Procedure{markLocalizable}{$x_i$} \hfill
\Comment{$x_i$ is the current ($t_i\in {\cal V}$ or $s_i\in V$)}
    \cmd{Wake}{on arrival of an elementary bar list, say $barLst$}
    \If{($barLst \cap \mt{elementLst} \neq \emptyset$ and $\mt{status}\neq
                        \mathbf{localizable}$)}
        \cmd{Set}{$\mt{status}\gets \mathbf{localizable}$}
        \cmd{Send}{a message with $\mt{elementLst}$ to all nodes in 
            $x_i.\mt{children}$}
    \EndIf
\EndProcedure
\end{algorithmic}
\end{algorithm}
If the current node, $x_i$, lies in an elementary bar in the received list and
is not marked as localizable, the process marks $x_i$ as localizable and sends
$x_i.\mt{elementLst}$ to all nodes in $x_i.\mt{children}$.
\begin{theorem} \label{Th:locTriangleBar}
If an elementary bar contains $T_1$, then all the nodes (in $G$) of these bars
are marked as localizable through \Call{markLocalizable}{\,} (identifying
the bar through $T_1$) in the complete run of the algorithm.
\end{theorem}

\begin{proof}
The result follows from the statements in
Proposition~\ref{prop:triangleCycleToCycleInFTG},
\ref{prop:triangleTreeToTreeInFTG}, \ref{prop:allTriangleCyces}
\ref{prop:allCircuitBrdgeNet}.
\qed \end{proof}

\section{Performance analysis of the algorithm}
\label{sec:performanceAnalysis}
%
%
\subsubsection{Synchronization, Correctness and Progress:}
Phase I computes the $FTG(G)$. Its synchronization, progress, finite termination and correctness of 
are proved in Proposition~\ref{prop:createFTG}. 
Phase II uses $BFS$ tool to find the trees in a connected $FTG(G)$. This ensures that the 
$FTT$s are found correctly~(Proposition~\ref{prop:maxTriangleTreeToFTT}).
Theorems~\ref{Th:triangleBar} and \ref{Th:locTriangleBar} establish the correctness of Phase II 
and III of the 
algorithm.

Each of the procedures in the algorithm may contain loop which run over a list either $\mt{nbrs}$ or 
$\mt{trngls}$. The sizes of these lists do not exceed the number of neighbours of a node. The loops 
executes without waiting for any signal. Therefore, the finite termination of each procedure in any 
individual node is guaranteed. Since the transmission medium is reliable, every message sometime 
reaches the destination. Thus, executions in the whole system terminate in finite time.


\subsubsection{Time complexity:}
We have used Lamport's logical clock. The maximum value of this logical clock in any node does not 
exceed thrice the number of its neighbours (Proposition~\ref{prop:progressPhaseI}). 
Phase I communicates no message beyond $2$-hop.
The running time complexity of \Call{recvNbrList}{\,} in each node is $O(n)$ in the worst 
case. Therefore, the worst case time complexity of Phase I is $O(n^2)$ in total. 
Time complexity for communication in Phase II and III is guided by the $BFS$ of $FTG(G)$ in 
distributed way. A visit
signal from $t_1$ will reach any other $t_i$ in $FTG(G)$ along its shortest path
in between them. In worst case, it may be equal $|V|$. This dominates
time required to find extended knots. Thus, the worst case time complexity is
equal to the number of nodes in $G$, i.e., $O(n)$. Thus the overall time complexity of the 
execution in the whole system is $O(n^2)$. 

\subsubsection{Energy complexity:}
Since message communication dominates the leading consumer of energy, we only
count the communications for energy analysis. Every node sends the neighbour list $\mt{nbrs}$ only 
once. It counts $n$ (number of nodes in the network) transmissions. A node executes 
\Call{recvNbrList}{\,} once for each of its neighbours. \Call{recvNbrList}{\,} sends a triangle 
message for a newly obtained triangle whose leader is a different. The total number of these 
triangle messages is $|E|$ maximum. \Call{recvNbrList}{\,} also sends a 
\textbf{flip} message for a newly obtained flip with a triangle whose leader is a different node. 
The total number of such \textbf{flip} messages does not exceed $|E|$. Thus the number message 
transmissions in Phase I is $O(|E|)$. It is easy to see that, in Phase II and Phase III, each node
communicates with its neighbours constant number of times. Hence, the total energy dissipation is 
$O(|E|)$ in worst.

\section{Conclusion} \label{conclude}
In this paper, we consider the problem of localizability of nodes as well as
networks. We do not compute the positions of nodes. So, exact distances are not
necessary and error in distance measurements does not affect the localizability
testing. We propose an efficient distributed technique to solve this problem for
a specific class of networks, triangle bar. The proposed technique is better
than both trilateration and wheel extension techniques. We also illustrate some
network scenarios which are wrongly reported as not localizable by wheel
extension technique, but the proposed algorithm recognizes them as localizable 
in a distributed environment. The proposed algorithm runs with $O(|V|^3)$ time
complexity and energy complexity of $O(|E|)$ in the worst case.

In centralized environment localizability testing can be carried out in
polynomial time. Though the proposed algorithm recognizes a class of localizable
networks distributedly, several localizable networks remains unrecognized by
this technique.
For localizability testing, we only consider a maximal triangle bar of $G$ which
includes the anchor triangle $T_1$ in $G=(V,E,d)$. Our future plan is to extend
localizability testing considering other triangle bars in $FTG(G)$.

%

\bibliographystyle{./splncs} 
\bibliography{../../mybibfile} 

\newpage

\appendix
\section{Appendix}
\subsection{Proof of the first part of Lemma~\ref{lem:trngleCycleToCircuit}} 
\label{app:trngleCycleToCircuit}

\begin{proof}
A wheel graph $W_n$ with $n$ nodes is a triangle cycle with $n-1$ triangles. A triangle 
cycle with three or four triangles is a wheel $W_4$  or a wheel $W_5$ 
(Fig.~\ref{fig:trngleCycle3T4T}) respectively. These wheels show the existence of spanning wheels 
for some triangle cycles. Consider a sufficiently large triangle cycle $G({\cal T})$ with the 
triangle stream ${\cal T}=(T_1, T_2, \ldots, T_n)$ which has no spanning wheel. If we consider 
three 
consecutive triangles in ${\cal T}$, then we have at least one node with degree at least four. 
Consider such a node $v$ with $deg(v)\geq 4$. If $deg(v) = 4$, then the edges incident on $v$ lie 
in 
three consecutive triangles $T_i$, $T_{i+1}$, $T_{i+2}$ (Fig.~\ref{fig:trngleCycleToCircuit}\,(a)).
\begin{figure}[h]
\begin{minipage}[t]{.3\textwidth}
    \centering
    \includegraphics[height=.5in]{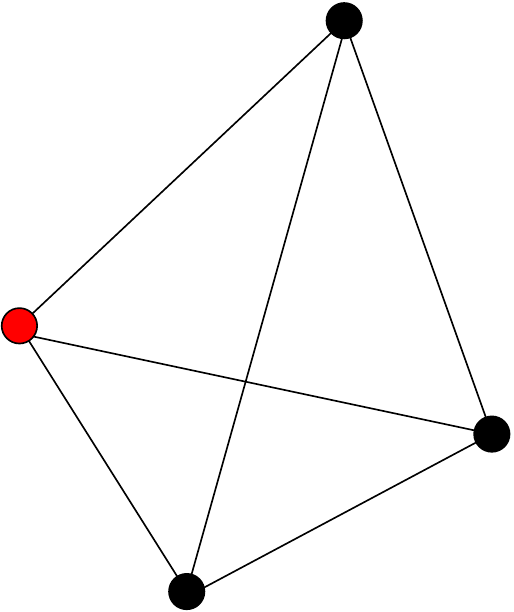} ~~~~~~
    \includegraphics[height=.5in]{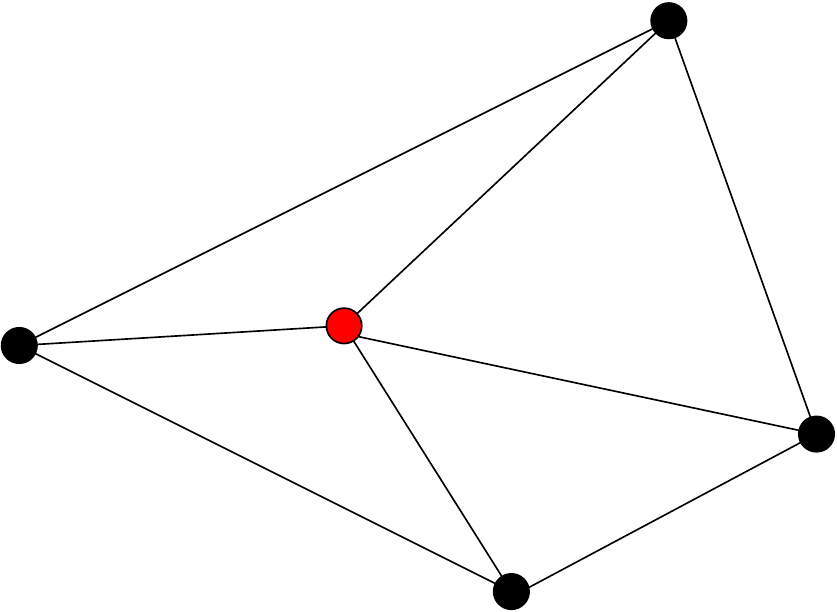}
\caption{Triangle cycle with three  and four triangles}
\label{fig:trngleCycle3T4T}
\end{minipage}
\begin{minipage}[t]{.6\textwidth}
    \centering
    \includegraphics[height=1in]{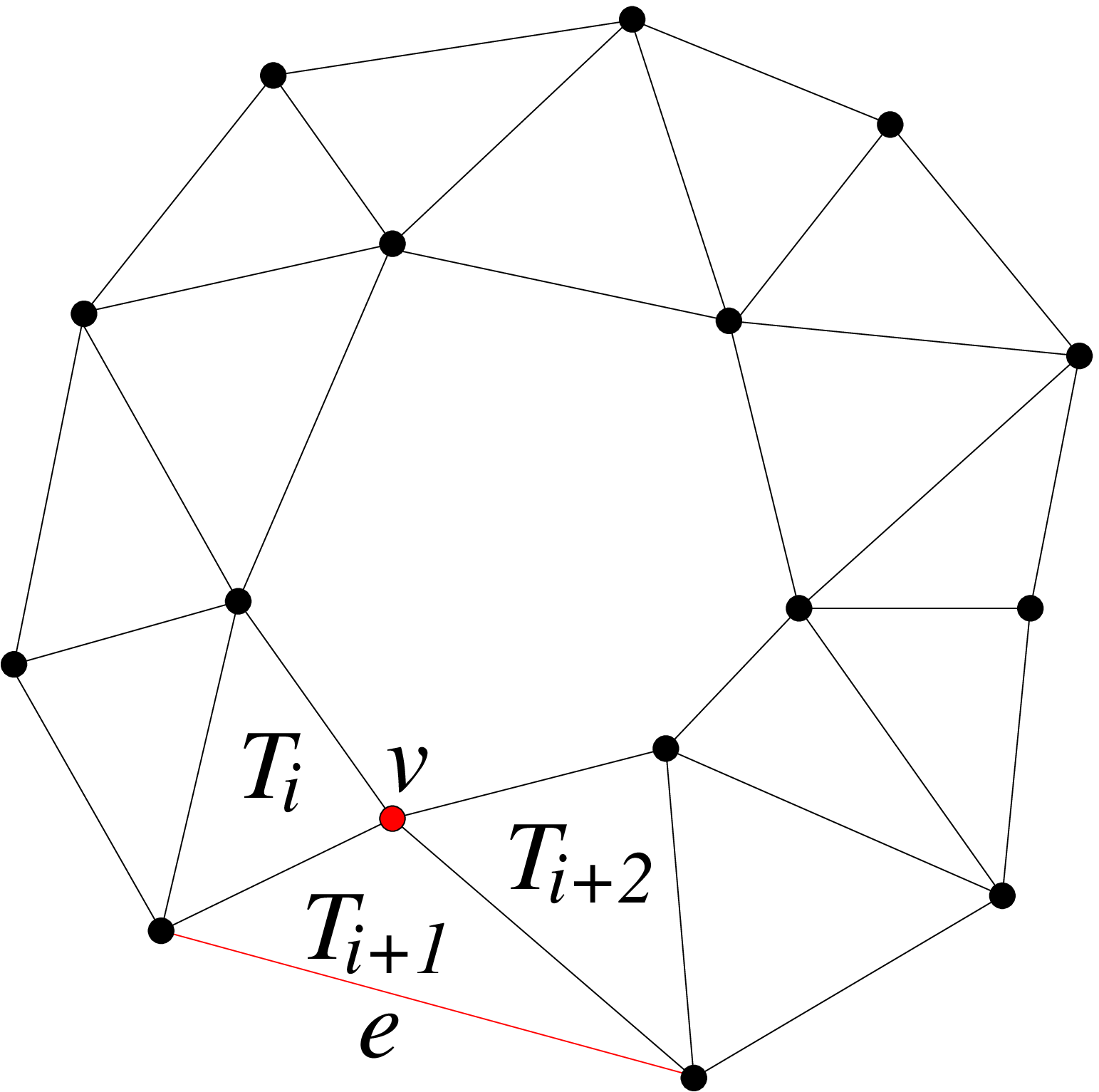} ~~~~~~~~~~~~~
    \includegraphics[height=1in]{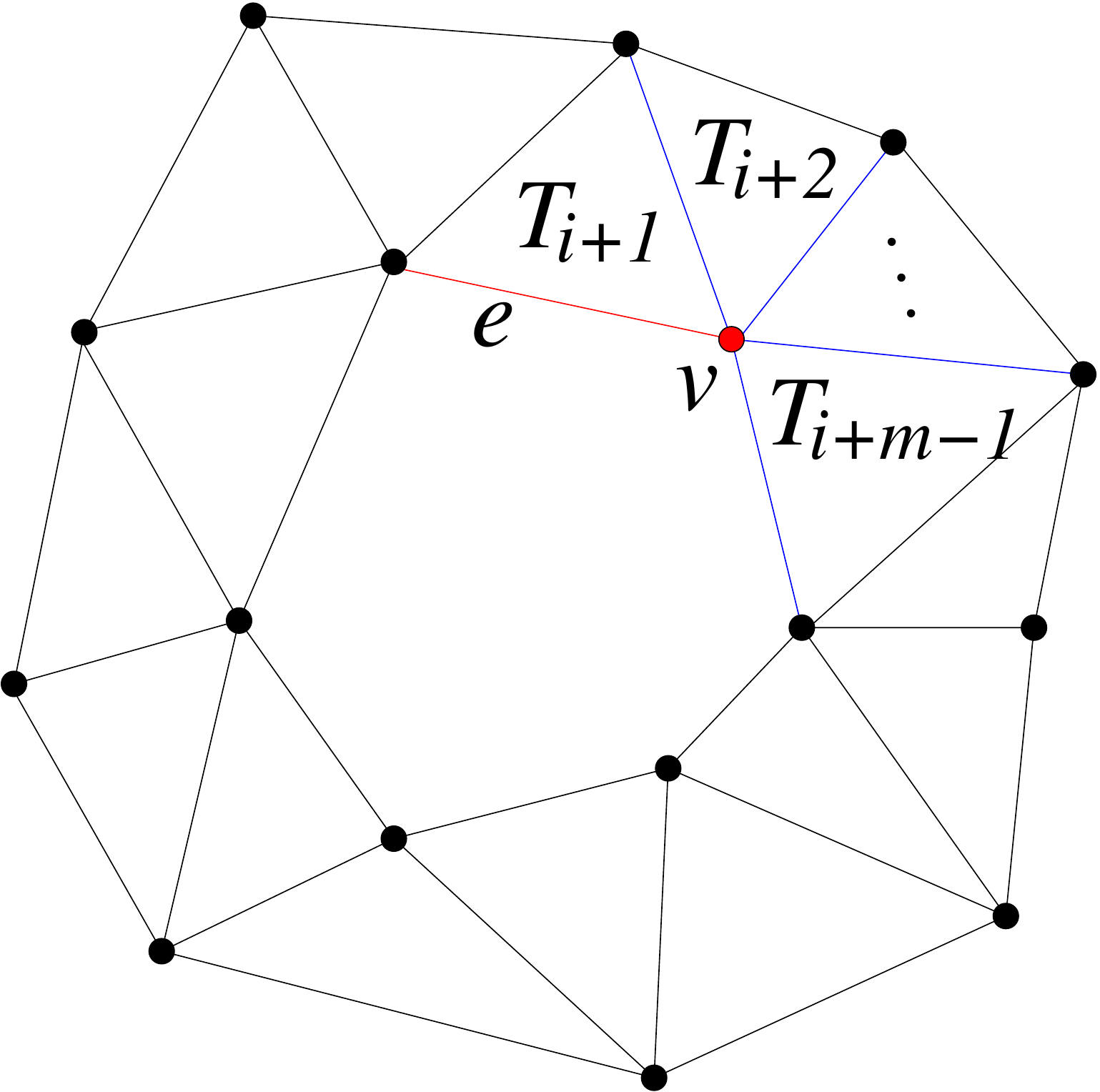}\\
~~~~~~~~~~~~~~~~~(a) ~~~~~~~~~~~~~~~~~~~~~~~~~~~~~~~~~~ (b) ~~~~~~~~~~~~~~
\caption{Triangle cycles without spanning wheel}
\label{fig:trngleCycleToCircuit}
\end{minipage}
\end{figure}\vspace*{-4mm}
Deleting the outer side $e$ of $T_{i+1}$ (i.e., $G({\cal T})-e$) gives a spanning triangle circuit
of $G({\cal T})$. Suppose, $deg(v) = m > 4$. The edges adjacent to $v$ lie in $m-1$ consecutive
triangles $T_{i+1}$, $T_{i+2}$, $\ldots$, $T_{i+m-1}$ (Fig.~\ref{fig:trngleCycleToCircuit}\,(b)).
Deleting the outer side of $T_{i+1}$ gives a spanning triangle circuit of $G({\cal T})$.
\qed \end{proof}

\subsection{Proof of the second part of Lemma~\ref{lem:trngleCircuitToBridge}} 
\label{app:trngleCircuitToBridge}

\begin{proof}
Let $G({\cal T})$ be a triangle circuit with triangle stream ${\cal T}=(T_1$, $T_2$, $\ldots$,
$T_n)$ and circuit knot $v$ (Fig.~\ref{fig:trngleCircuitToBridge}). $v$ is the only pendant for
both $T_1$ and $T_n$. $T_2$ and $T_3$ share the edge $f$ and  $w$ is the
corresponding pendant in $T_2$.
\begin{figure}[h]
\begin{minipage}[c]{2.5in}
    \centering
    \includegraphics[height=.9in]{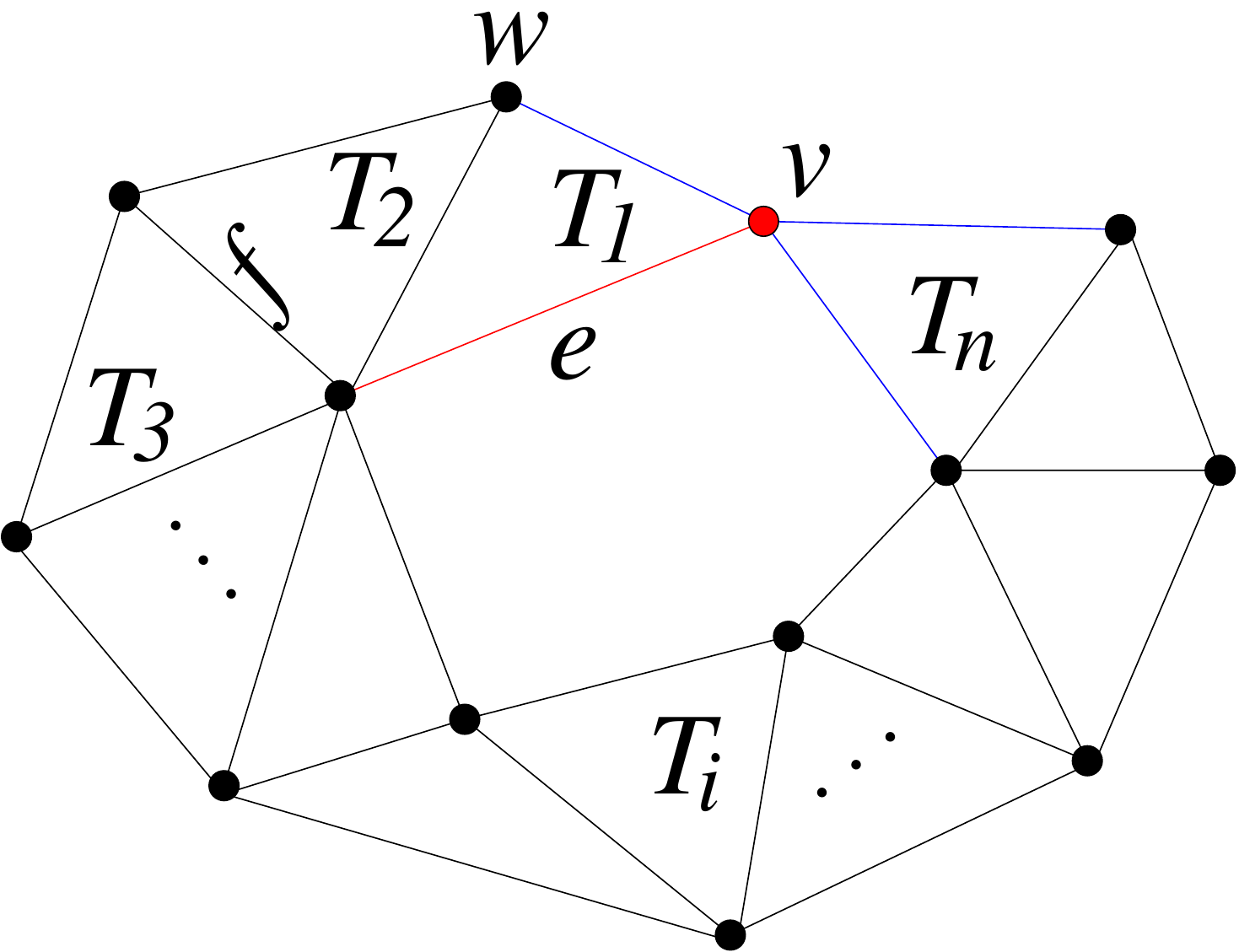} 
\end{minipage}
\begin{minipage}[c]{3in}
\caption{Triangle circuit ${\cal T}$ gives a spanning triangle bridge ${\cal T}-e$}
\label{fig:trngleCircuitToBridge}
\end{minipage}
\end{figure}\vspace*{-5mm}
$T_1$ has two outer sides which are incident on $v$; one edge is incident on $w$ and the other is
incident on $f$. If $e$ is the outer side of $T_1$ incident on $f$, then  $G({\cal T})-e$ gives a
spanning triangle bridge of $G({\cal T})$.
\qed \end{proof}

\subsection{Proof of Lemma~\ref{lem:triangleNotch}} \label{app:triangleNotch}
\begin{proof} 
Let $G$ be a triangle notch generated from a triangle tree $G({\cal T})$ with the apex $v$. $G$
contains $n+2$ nodes $u_i$, $i=1$, $2$, $\ldots$, $n+2$. All the leaf knots of $G({\cal T})$ are
adjacent to $v$. Fig.~\ref{fig:triangleNotch}\,(b) shows an example of such a graph. The leaf
knots $u_i$, $u_j$ and $u_k$ of $G({\cal T})$ are adjacent to $v$.
Consider a generic configuration $\cal P$ of $G$ where $v$ is realized as $(x,y)$
and $u_i$ as $(x_i,y_i)$ for $i=1$, $2$, $\ldots$, $n+2$. $G({\cal T})$ can have only flips. If 
possible, let a flip operation on $\cal P$ generate a different configuration ${\cal P}'$ with 
coordinates $(x',y')$ for $v$ and $(x_i',y_i')$ for
$u_i$, $i=1$, $2$, $\ldots$, $n+2$. Without loss of generality, we assume that the leaf triangle
$T_i$, with a leaf knot $u_i$, remains fixed in both the configurations $\cal P$ and ${\cal P}'$.
Consider another leaf triangle $T_j$ with leaf knot $u_j$. Let ${\cal T}_{ij}$ be the unique
triangle stream from $T_i$ to $T_j$ in the graph $G({\cal T})$. Proceeding in a manner similar to
that in the proof of Lemma~\ref{lem:triangleCycleCircuitBridge}, each of $x_j'$ and $y_j'$ can be
expressed in the form of $\frac{\phi}{\psi}$ where $\phi$ and $\psi$ are two non-zero polynomials of
the coordinates of the points in ${\cal T}_{ij}-\{u_i\}$ with integer coefficients such that
$\psi\neq 0$. From elementary coordinate geometry in $\mathbb R ^2$, $x'$ and $y'$ can also be
expressed similarly in terms of the coordinates of the points in ${\cal T}_{ij}$.
 Similarly, from the triangle stream ${\cal T}_{jk}$, $x_k'$ and $y_k'$ have similar expressions
involving the coordinates of the points in ${\cal T}_{jk}-\{u_i,v\}$. Since, the edge distance
between $u_k$ and $v$ is given, then at least the coordinates of $u_i$, $v$ and $u_k$ are
algebraically dependent. This contradicts the assumption that every three nodes in $\cal P$ are in
general position. So the union of ${\cal T}_{ij}$, ${\cal T}_{jk}$ and $v$ in $\cal P$ admits no
flip; and the union is generically globally rigid. Since, the leaf triangles chosen are arbitrary
and any triangle lies on at least one triangle stream between some pair of leaf triangles, the
generic global rigidity of $G$ follows.
\qed \end{proof}

\subsection{Proof of Lemma~\ref{lem:triangleNet}} \label{app:triangleNet}
\begin{proof}
Consider a triangle net $G$ generated by a triangle tree $G({\cal T})$ where ${\cal T} = (T_1, T_2,
\ldots, T_n)$. If $u_r$ is an apex of $G$, then the extended nodes have an ordering $u_1$, $u_2$,
$\ldots$, $u_r$. Consider a generic configuration of $G$ (e.g.
Fig.~\ref{fig:triangleNetProof}).
\begin{figure}[h]
\begin{minipage}[c]{.3\textwidth}
    \caption{Triangle nets with extended nodes $u$ and $v$ where (a) $u$, $v$ are adjacent to
        pendants only,         (b) $u$, $v$ are adjacent to both pendant and extended nodes}
    \label{fig:triangleNetProof}
\end{minipage}
\begin{minipage}[c]{.7\textwidth}
    \centering
    \includegraphics[height=1in]{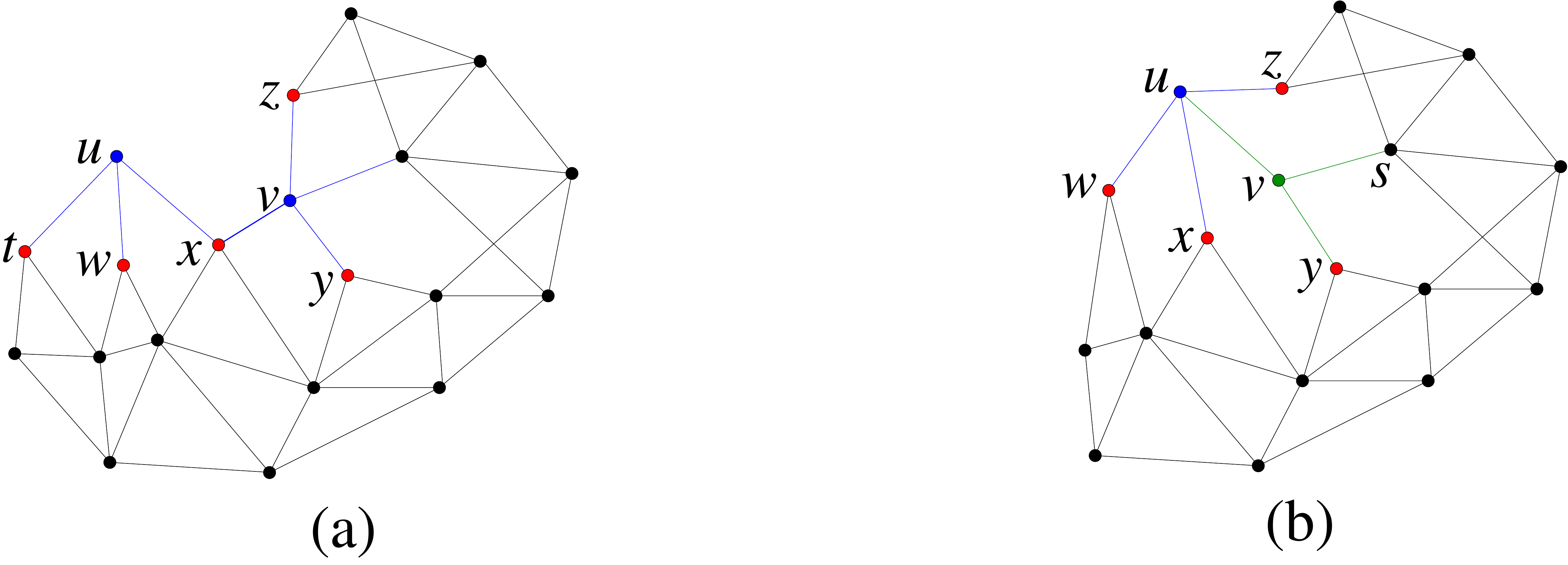}
\end{minipage}
\end{figure}\vspace*{-4mm}
 The extended node added to $G({\cal T})$ is $u_1$ which is adjacent to only pendants. These 
pendants are leaf knots of a triangle tree, which is a subgraph of $G({\cal T})$.
This triangle tree is a triangle notch; hence it is generically globally rigid
(Lemma~\ref{lem:triangleNotch}). 

Consider a leaf knot $y$ (e.g., Fig.~\ref{fig:triangleNetProof}). Let $P$ be an extending
path which connects $y$ to the apex $u_r$. In view of Lemma~\ref{lem:extendedNode}, considering the
extended nodes along $P$ and combining the corresponding generically globally rigid graphs, 
$G$ is generically globally rigid.
\qed \end{proof}

\subsection{Proof of Theorem~\ref{Th:trilaterionWhlextnInBar}} \label{app:trilaterionWhlextnInBar}

\begin{proof}
Let $G$ be a trilateration graph having a trilateration ordering $\pi = (u_1, u_2, \ldots, u_n)$
where $u_1$, $u_2$ and $u_3$ are in $K_3$. $u_4$ is adjacent to three nodes before $u_4$ in
$\pi$. So $u_1$, $u_2$, $u_3$ and $u_4$ form a $K_4$ which is a triangle cycle and consequently
a triangle bar. Suppose $\pi' =(u_1, u_2, \ldots, u_i)$, $4 \leq i< n$ forms a triangle bar ${\cal
B}'$. The node $u_{i+1}$ is adjacent to at least three nodes in ${\cal B}'$. Therefore, ${\cal
B}'\cup \{v_{i+1}\}$ is a triangle bar. By mathematical induction, $G$ is a triangle bar.

Consider a wheel extension graph $G$ with a node ordering $\pi = (u_1$, $u_2$, $\ldots$, $u_n)$.
$u_1$, $u_2$ and $u_3$ are in $K_3$ and $u_i$ ($i\geq 4$) lies in a wheel containing at least three
nodes in $\pi$ before $u_i$. So, $u_4$ lies on a wheel, say $W_1$, which contains $u_1$, $u_2$
and $u_3$. If any, let $u_j$, $j>4$, be the first node in $\pi$ such that $u_j$ does not lie on
$W_1$. $u_j$ lies on another wheel, say $W_2$, which shares at least three nodes with $W_1$. 
Therefore, $W_1\cup W_2$ is generically globally rigid (by Lemma~\ref{lem:3commonPoint}).
Similarly, let $u_k$, if any, be the first node in $\pi$ such that $u_k$ does not lie on $W_1\cup
W_2$. Assume $u_k$ lies on a wheel $W_3$ which shares three nodes with $W_1\cup W_2$. By
Lemma~\ref{lem:3commonPoint}, $W_1\cup W_2\cup W_3$ is generically globally rigid. Proceeding in
this way, we can obtain, ${\cal W} = (W_1, W_2, \ldots, W_m)$, a finite sequence of wheels such
that each $W_i$ shares at least three nodes on some $W_j$s before $W_i$ in ${\cal W}$ and  $G
= \bigcup\limits_{i =1}^m W_i$, $m\geq 1$. Wheel graph is triangle cycle. Hence, $G$ is a
triangle bar.
\qed \end{proof}

\end{document}